\def \IsDraft{}

\documentclass[11pt]{article}

\usepackage{color}
\usepackage{amsfonts,amsmath,amssymb,amsthm,boxedminipage,color,url,fullpage}
\definecolor{weborange}{rgb}{.8,.3,.3}
\definecolor{webblue}{rgb}{0,0,.8}
\definecolor{internallinkcolor}{rgb}{0,.5,0}
\definecolor{externallinkcolor}{rgb}{0,0,.5}

\usepackage[pagebackref,
hyperfootnotes=false,
colorlinks=true,
urlcolor=externallinkcolor,
linkcolor=internallinkcolor,
filecolor=externallinkcolor,
citecolor=internallinkcolor,
breaklinks=true,
pdfstartview=FitH,
pdfpagelayout=OneColumn]{hyperref}

\usepackage{amsthm} 

\usepackage{comment}

\usepackage{enumerate,paralist}
\usepackage[labelfont=bf]{caption}

\usepackage{aliascnt}
\usepackage[numbers]{natbib} 
\usepackage{cleveref}
\usepackage{xspace}
\usepackage{xstring}

\usepackage{tikz}\usetikzlibrary{arrows}
\usetikzlibrary{decorations.markings}
\usepackage{subcaption}

\usepackage{tabularx}
\usepackage{bbold}
\newcommand{\remove}[1]{}

\newcommand{\TLLNCS}[2]{\ifdefined\IsLLNCS#1\else #2 \fi}

\ifdefined\IsDraft
    \newcommand{\authnote}[2]{{\bf [{\color{red} #1's Note:} {\color{blue} #2}]}}
      
\else
    \newcommand{\authnote}[2]{}
    
\fi

\ifdefined\IsDraft
    
\else
    
\fi

\ifdefined\IsDraft
    \newcommand{\deleted}[1]{{\color{blue} ~Deleted:~{\color{red} #1}}}
\else
    \newcommand{\deleted}[1]{}
\fi

\newenvironment{protocol}{\begin{mybox} \vspace{-.1in}\begin{proto}}{ \vspace{-.1in} \end{proto}\end{mybox}}

\newenvironment{algorithm}{\begin{mybox} \vspace{-.1in}\begin{algo}}{ \vspace{-.1in} \end{algo}\end{mybox}}

\newenvironment{mybox}{\begin{center}\begin{tabular}{|p{0.97\linewidth}|c|}   \hline} {  \\ \hline \end{tabular} \end{center}}

\newcommand{\1}{\mathbb{1}}

\newcommand{\Ensuremath}[1]{\ensuremath{#1}\xspace}
\newcommand{\MathAlg}[1]{\mathsf{#1}}

\newcommand{\MathAlgX}[1]{\Ensuremath{\MathAlg{#1}}}

\newcommand{\aka} {also known as,\xspace}
\newcommand{\resp}{resp.,\xspace}
\newcommand{\ie}  {i.e.,\xspace}
\newcommand{\eg}  {e.g.,\xspace}

\newcommand{\wrt} {with respect to\xspace}
\newcommand{\wlg} {without loss of generality\xspace}

\newcommand{\set}[1]{\ens{#1}}

\newcommand{\half}{\tfrac{1}{2}}

\newcommand{\R}{{\mathbb R}}
\newcommand{\N}{{\mathbb{N}}}

\newcommand{\io}{\class{i.o. \negmedspace-}}
\newcommand{\zo}{\set{0,1}}

\newcommand{\zs}{{\zo^\ast}}

\newcommand{\xor}{\oplus}

\newcommand{\eps}{\varepsilon}

\newcommand{\from}{\leftarrow}
\newcommand{\la}{\gets}

\newcommand{\poly}{\operatorname{poly}}

\newcommand{\Exp}{\Ex}

\newcommand{\negl}{\operatorname{neg}}

\newcommand{\Supp}{\operatorname{Supp}}

\renewcommand{\cref}{\Cref}

\TLLNCS{
\newaliascnt{claiml}{theorem}
\newtheorem{claiml}[claiml]{Claim}
\aliascntresetthe{claiml}

\renewenvironment{claim}{\begin{claiml}}{\end{claiml}}
}
{
\newtheorem{theorem}{Theorem}[section]

\newaliascnt{lemma}{theorem}
\newtheorem{lemma}[lemma]{Lemma}
\aliascntresetthe{lemma}

\newaliascnt{claim}{theorem}
\newtheorem{claim}[claim]{Claim}
\aliascntresetthe{claim}

\newaliascnt{corollary}{theorem}

\aliascntresetthe{corollary}

\newaliascnt{proposition}{theorem}
\newtheorem{proposition}[proposition]{Proposition}
\aliascntresetthe{proposition}

\newaliascnt{conjecture}{theorem}

\aliascntresetthe{conjecture}

\newaliascnt{definition}{theorem}
\newtheorem{definition}[definition]{Definition}
\aliascntresetthe{definition}

\newaliascnt{remark}{theorem}
\newtheorem{remark}[remark]{Remark}
\aliascntresetthe{remark}

\newaliascnt{example}{theorem}

\aliascntresetthe{example}
}

\crefname{lemma}{Lemma}{Lemmas}
\crefname{figure}{Figure}{Figures}
\crefname{claim}{Claim}{Claims}
\crefname{corollary}{Corollary}{Corollaries}
\crefname{proposition}{Proposition}{Propositions}
\crefname{conjecture}{Conjecture}{Conjectures}
\crefname{definition}{Definition}{Definitions}
\crefname{remark}{Remark}{Remarks}
\crefname{exmaple}{Example}{Examples}

\newaliascnt{construction}{theorem}

\aliascntresetthe{construction}
\crefname{construction}{Construction}{Constructions}

\newaliascnt{fact}{theorem}

\aliascntresetthe{fact}
\crefname{fact}{Fact}{Facts}

\newaliascnt{notation}{theorem}
\newtheorem{notation}[notation]{Notation}
\aliascntresetthe{notation}
\crefname{notation}{Notation}{Notation}

\crefname{equation}{Equation}{Equations}

\newaliascnt{proto}{theorem}

\newtheorem{proto}[proto]{Protocol}

\aliascntresetthe{proto}
\crefname{proto}{protocol}{protocols}

\newaliascnt{algo}{theorem}
\newtheorem{algo}[algo]{Algorithm}
\aliascntresetthe{algo}
\crefname{algo}{algorithm}{algorithms}

\newaliascnt{expr}{theorem}
\newtheorem{expr}[expr]{Experiment}
\aliascntresetthe{expr}
\crefname{experiment}{experiment}{experiments}

\newcommand{\stepref}[1]{Step~\ref{#1}}

\def\FullBox{$\Box$}
\def\qed{\ifmmode\qquad\FullBox\else{\unskip\nobreak\hfil
\penalty50\hskip1em\null\nobreak\hfil\FullBox
\parfillskip=0pt\finalhyphendemerits=0\endgraf}\fi}

\def\qedsketch{\ifmmode\Box\else{\unskip\nobreak\hfil
\penalty50\hskip1em\null\nobreak\hfil$\Box$
\parfillskip=0pt\finalhyphendemerits=0\endgraf}\fi}

\newcommand{\eex}[2]{\Ex_{#1}\left[#2\right]}
\newcommand{\ex}[1]{\Ex\left[#1\right]}
\newcommand{\Ex}{{\mathrm E}}
\renewcommand{\Pr}{{\mathrm {Pr}}}
\newcommand{\pr}[1]{\Pr\left[#1\right]}
\newcommand{\ppr}[2]{\Pr_{#1}\left[#2\right]}

\newcommand{\Ac}{\MathAlgX{A}}
\newcommand{\Ah}{\MathAlgX{\widehat{A}}}
\newcommand{\At}{\MathAlgX{\widetilde{A}}}
\newcommand{\Ao}{\MathAlgX{\overline{A}}}
\newcommand{\Inv}{\MathAlgX{Inv}}
\newcommand{\Bh}{\MathAlgX{\widehat{B}}}
\newcommand{\Bt}{\MathAlgX{\widetilde{B}}}
\newcommand{\Bo}{\MathAlgX{\overline{B}}}

\newcommand{\Po}{\MathAlgX{\overline{\pi}}}
\newcommand{\Ph}{\MathAlgX{\widehat{\pi}}}
\newcommand{\Pt}{\MathAlgX{\widetilde{\pi}}}

\newcommand{\Xo}{\overline{X}}

\newcommand{\Xt}{\widetilde{X}}

\newcommand{\Yo}{\overline{Y}}

\newcommand{\Yt}{\widetilde{Y}}

\newcommand{\To}{\overline{T}}

\newcommand{\Tt}{\widetilde{T}}

\newcommand{\Zo}{\overline{Z}}

\newcommand{\Zt}{\widetilde{Z}}

\newcommand{\Ot}{\widetilde{O}}

\newcommand{\Dc}{\MathAlgX{D}}

\newcommand{\Pc}{\MathAlgX{P}}

\newcommand{\Bc}{\MathAlgX{B}}
\newcommand{\Cc}{\MathAlgX{C}}

\newcommand{\ens}[1]{\left\{#1\right\}}
\newcommand{\size}[1]{\left|#1\right|}

\newcommand{\cindist}{\mathbin{\stackrel{\rm C}{\approx}}}

\newcommand{\out}{\operatorname{out}}

\newcommand{\trans}{{\operatorname{trans}}}

\newcommand{\view}{\operatorname{view}}

\newcommand{\Uni}{{\mathord{\mathcal{U}}}}

\newcommand{\prob}[1]{\mathsf{\textsc{#1}}}

\newcommand{\SD}{\prob{SD}}

\newcommand{\I}{\mathcal{I}}

\newcommand{\ppt}{{\sc ppt}\xspace}
\newcommand{\pptm}{{\sc pptm}\xspace}

\newcommand{\cS}{\mathcal{S}}

\newcommand{\cG}{\mathcal{G}}
\newcommand{\cF}{\mathcal{F}}

\newcommand{\p}{{{\mathsf{P}}}}
\newcommand{\D}{{{\mathsf{D}}}}

\newcommand{\cs}{{\cal{S}}}

\newcommand{\is}{{i^\ast}}

\newcommand{\Tableofcontents}{
\ifdefined\IsLLNCS \else
\thispagestyle{empty}
\pagenumbering{gobble}
\clearpage
\ifdefined\IsSubmission \else
\tableofcontents
\thispagestyle{empty}
\clearpage
\fi
\pagenumbering{arabic}
\fi
}

\newcommand{\party}[1]{%
    \IfEqCase{#1}{%
        {1}{\Ac}
        {2}{\Bc}
        {3}{\Cc}
    }[\PackageError{\party}{Undefined option to party: #1}{}]%
}%

\newcommand{\secParam}{\kappa}

\newcommand{\SMbox}[1]{\mbox{\scriptsize {\sc #1}}}
\newcommand{\REAL}{\SMbox{REAL}}

\mathchardef\mhyphen="2D

\newcommand{\remph}[1]{\textsf{#1}}

\newcommand{\XOR}{\text{DP-XOR}}

\newcommand{\Ec}{\MathAlgX{E}}

\newcommand{\Fc}{\MathAlgX{F}}

\newcommand{\pk}{\secParam}  

\newcommand{\II}{\mathcal{I}}

\newcommand{\tth}[1]{#1\Ensuremath{^{\rm th}}}
\newcommand{\ith}{\tth{i}}

\newcommand{\SIML}{\SMbox{SML}}
\newcommand{\FSC}{\SMbox{FST}}
\newcommand{\PFSC}{\SMbox{PFST}}
\newcommand{\UCR}{\SMbox{UCR}}

\newcommand{\Price}{\mathsf{price}}

\newcommand{\Decr}{\MathAlgX{Decor}}

\newcommand{\KA}{\Phi}

\newcommand{\sX}{X'}
\newcommand{\sY}{Y'}

\newcommand{\sA}{{\widehat{\Ac}}}
\newcommand{\sB}{{\widehat{\Bc}}}

\newcommand{\Fh}{\MathAlgX{\widehat{\Fc}}}

\newcommand{\Prod}{{\mathsf{prod}}}
\newcommand{\Sim}{\MathAlgX{Sim}}

\newcommand{\INFO}{{\mathsf{IT}}}
\newcommand{\COMP}{{\mathsf{COM}}}

\newcommand{\pb}{p_{\Bc}}
\newcommand{\pa}{p_{\Ac}}

\newcommand{\pbb}[1]{p_{\Bc| #1}}

\newcommand{\pbz}{\pbb{0}}
\newcommand{\pbo}{\pbb{1}}

\remove{

}

\title{Computational Two-Party Correlation: \\A Dichotomy for Key-Agreement Protocols\thanks{The full version was published in the SIAM Journal on Computing 2020 \cite{SIAM2020}. An extended abstract of this work appeared in  the Annual Symposium on Foundations of Computer Science (FOCS)  2018 \cite{HNOSS18}.}
}

\author{Iftach Haitner\thanks{School of Computer Science, Tel Aviv University. E-mail: \texttt{iftachh@cs.tau.ac.il}. Member of the Check Point Institute for Information Security. Research supported by ERC starting grant 638121.}
\and Kobbi Nissim\thanks{Department of Computer Science, Georgetown University. E-mail: \texttt{kobbi.nissim@georgetown.edu}. Research supported by NSF grant CNS-1565387.}
\and Eran Omri\thanks{Department of Computer Science, Ariel University. E-mail: \texttt{omrier@ariel.ac.il}. Research supported by ISF grants 544/13 and 152/17.}
\and Ronen Shaltiel\thanks{Department of Computer Science. University of Haifa, E-mails: \texttt{ronen@cs.haifa.ac.il}. Research supported by ISF grant 1628/17.}
\and Jad Silbak\thanks{School of Computer Science, Tel Aviv University. E-mail: \texttt{jadsilbak@mail.tau.ac.il}. Research supported by ISF grant 1628/17 and by ERC starting grant 638121.}
}

\date{\vspace{-3ex}}

\begin{document}
\sloppy

\maketitle

\vspace{-0.1in}
\begin{abstract}
Let $\pi$ be an efficient two-party protocol that given security parameter $\kappa$, both parties output single bits $X_\kappa$ and $Y_\kappa$, respectively. We are interested in how $(X_\kappa,Y_\kappa)$ ``appears'' to an  efficient adversary  that only views the transcript $T_\kappa$. We make the following contributions:

\begin{itemize}
\item We develop new tools to argue about this loose notion and show (modulo some caveats) that for every such protocol $\pi$, there exists an efficient \textit{simulator} such that the following holds:  on input $T_\kappa$, the simulator outputs  a pair $(X'_\kappa,Y'_\kappa)$ such that $(X'_\kappa,Y'_\kappa,T_\kappa)$  is (somewhat) \emph{computationally indistinguishable} from  $(X_\kappa,Y_\kappa,T_\kappa)$.

\item We use these tools to prove the following \emph{dichotomy theorem}: every such  protocol $\pi$ is:
    \begin{itemize}
    \item either  \textit{uncorrelated} ---   it is (somewhat) indistinguishable from an efficient   protocol whose parties interact to produce $T_\kappa$, but then choose their outputs \emph{independently} from some product distribution (that is determined in poly-time from $T_\kappa$),

    \item or, the protocol implies a key-agreement protocol (for infinitely many $\kappa$'s). 
    \end{itemize}

    Uncorrelated protocols are  uninteresting from a cryptographic viewpoint, as the correlation between outputs is (computationally) trivial. Our dichotomy shows that every protocol is either completely uninteresting or implies key-agreement.

\item We use the above  dichotomy  to make progress on open problems on minimal cryptographic assumptions required for differentially private mechanisms for the XOR function.

\item A subsequent work of \citeauthor{HMO}  uses the above dichotomy  to makes progress on a long-standing open question regarding the  complexity  of fair two-party  coin-flipping protocols.
\end{itemize}

\noindent
We highlight the following two ideas regarding our technique:
\begin{itemize}
\item The simulator algorithm  is obtained by a carefully designed ``competition'' between efficient  algorithms attempting to forecast $(X_\kappa,Y_\kappa)|_{T_\kappa=t}$. The winner is used to simulate the outputs of the protocol.
\item Our key-agreement protocol uses the simulation to reduce to an information theoretic setup, and is in some sense non-black box.
\end{itemize}
\end{abstract}



\remove{
\thispagestyle{empty}
\pagenumbering{gobble}
\clearpage
\pagenumbering{arabic}
}
\Tableofcontents

\section{Introduction}\label{sec:intro}

In this paper, we discuss ``computational correlation" of efficient single-bit output  two-party protocols. We start with some notation for such protocols.

\paragraph{Two-party protocols with single bit output.} We are interested in probabilistic polynomial-time (\ppt), two-party, no-input,  single-bit output protocols: the  \ppt  parties  receive a common input $1^\kappa$ (\ie a security parameter), and each party outputs a single  bit. For such protocols $\pi=(\Ac,\Bc)$ we use the  notation:
\[ \pi(1^\kappa) = (\Ac,\Bc)(1^\kappa)=(X_\kappa,Y_\kappa,T_\kappa). \]
Where $X_\kappa$ is the output of $\Ac$, $Y_\kappa$ is the output of $\Bc$, and $T_\kappa$ is the transcript of the protocol.  Loosely speaking, we are interested in the correlation that an execution of $\pi(1^\kappa)$ generates between $X_\kappa$ and $Y_\kappa$, when viewed from the point of view of a \ppt algorithm that receives only the transcript $T_\kappa$ as input.

\paragraph{Key-agreement protocols.}
We will be interested in ``computational correlation'' between the outputs of a protocol. It is instructive to consider the  example of key-agreement protocols. The latter  are  \ppt protocols with the following properties:
\begin{description}
\item[Secrecy.] $\pr{\Ec(T_\kappa)=X_\kappa}\le \frac{1}{2} + s(\kappa)$ for  every \ppt algorithm (eavesdropper)  \Ec. (Here the standard choice for $s(\kappa)$ is a negligible function, but we will also consider versions where $s(\kappa)=s$ is a constant).

\item[Agreement.]  $\pr{X_\kappa=Y_\kappa} \ge \frac{1}{2} + a(\kappa)$. (Here the standard choice for $a(\kappa)$ is half minus a negligible function, but we will also consider versions where $a(\kappa)=a$ is a constant, and $a>s$).
\end{description}
The reader is referred to \cite{Holenstein06b} for a survey on key-agreement protocols. We remark that by \cite{Holenstein06b}, a key-agreement protocol for constants $s$ and $a$ with $s<a^2/10$, implies a full-fledged key-agreement protocol (\ie with the standard choices of agreement and secrecy).

\paragraph{Computational correlation.}
Loosely speaking, from the ``point of view'' of a \ppt algorithm \Ec that only sees  the transcript $t$ of a key-agreement protocol, the probability space $(X_\kappa,Y_\kappa)|_{T_\kappa=t}$ ``should look like'' $(R,R)$, for $R$ being a uniform bit (unknown to \Ec). This in contrast to the  view of an \emph{unbounded} \Ec: since for any protocol, and every transcript $t$, $(X_\kappa,Y_\kappa)|_{T_\kappa=t}$ is a product distribution.\footnote{In an information theoretic setup (without a dealer), if the views of parties $\Ac$ and $\Bc$ have no correlation (a product distribution), then even after the parties interact, conditioned on this interaction (transcript) the view of both parties remains a product distribution.}

An important contribution of this paper is developing tools to formalize the vague notion of ``computational correlation'' in a rigorous (and as we shall explain) useful way. Specifically, we show that (modulo some caveats and technicalities that we soon explain) for every single-bit output, two-party protocol, there exists a \ppt algorithm (simulator)  \Sim such that the following holds:    on input  $T_\kappa$, \Sim outputs  two bits (simulated outputs) $(X'_\kappa,Y'_\kappa)$ such that the simulated  experiment $(X'_\kappa,Y'_\kappa,T_\kappa)$ is computationally indistinguishable from (real) experiment $(X_\kappa,Y_\kappa,T_\kappa)$.

The simulated  experiment represents the ``best understanding'' that a \ppt can obtain on the real experiment. We find it quite surprising that such a clean notion exists. One could have expected that different \ppt's have ``different views'' or ``different understanding'' of the real execution, and it is impossible to come up with a \emph{single} simulated distribution that represents the ``collective understanding'' of all \ppt's. Loosely speaking, the above yields that such two-party protocols can be classified as follows:
\begin{itemize}
\item Protocols in which the simulated distribution $(X'_\kappa,Y'_\kappa,T_\kappa)$ has the property that $(X'_\kappa,Y'_\kappa)$ are independent, conditioned on every fixing of $T_\kappa$. We will call such protocols ``uncorrelated''.
\item Protocols in which the simulated distribution $(X'_\kappa,Y'_\kappa,T_\kappa)$ has the property that $(X'_\kappa,Y'_\kappa)$ are correlated given $T_\kappa$ (at least for some fixings of $T_\kappa$).
\end{itemize}

\paragraph{Uncorrelated protocols are cryptographically uninteresting.}
Uncorrelated protocols are uninteresting from a cryptographic viewpoint;  whenever we have such a protocol $\pi$, we can imagine that the parties use the following alternative trivial protocol $\widehat{\pi} =(\sA,\sB)$: party  $\sA$ samples a transcript $T_\kappa$ (on his own) and  sends $T_\kappa$ to $\sB$. Then each party samples its output (independently) by applying  the simulator for $\pi$ on $T_\kappa$.

As is often the case in simulation, if a \ppt  adversary $\Ec$  is able to perform some task (that is defined in terms of the original triplet $(X_\kappa,Y_\kappa,T_\kappa)$), then it achieves roughly the same success on the simulated triplet $(\sX_\kappa,\sY_\kappa,T_\kappa)$. Specifically, if $\pi$ is a key-agreement protocol, then $\widehat{\pi}$ is also a key-agreement protocol. The latter, however, is obviously false. This is because given $T_\kappa$, the adversary $\Ec$ can use the simulator to sample $X'_\kappa$ with probability that is at least as large as $\pr{X'_\kappa=Y'_\kappa}$. This means that in $\widehat{\pi}$ secrecy is less than agreement, ruling out any meaningful form of key-agreement.

\paragraph{Correlated protocols yield key-agreement.}
In this paper, we prove that (again, modulo some caveats and technicalities that we soon explain) if a protocol is correlated, then it can be transformed into a key-agreement protocol. This can be interpreted as the following dichotomy theorem:

\begin{quote}
Every \ppt single-bit output  two-party  protocol  is either uncorrelated (and is indistinguishable from a trivial and cryptographically uninteresting protocol), or it implies a key-agreement protocol.
\end{quote}

\smallskip
We find this quite surprising. Intuitively, key-agreement protocols and trivial protocols represent two extremes in the spectrum of two-party protocols, and one may expect that there are many interesting intermediate types in between the two extremes.\footnote{One illuminating ``intermediate setup'' is ``defective key-agreement protocols'' in which the agreement and secrecy properties above hold, but with $a<s$ (namely, agreement is smaller than secrecy, and this is not a cryptographically meaningful key-agreement). Such protocols can be uncorrelated (and trivial), but they can also be correlated, and thus, by our result, imply key-agreement. As we shall explain, this approach yields several new results, as in some cases it was previously unknown whether key-agreement protocols are implied, but it is possible to show that the protocol is not uncorrelated.}

\subsection{Our Results}

\subsubsection{Every two-party single bit output protocol has a simulator and a forecaster}

We show that every protocol has a \ppt \emph{simulator} that, seeing only the transcript, produces a simulated distribution  simulating the (real) output distribution of the protocol.

\begin{theorem}[Existence of \ppt simulators (informal)]\label{thm:intro:simulator}
Let $\pi=(\Ac,\Bc)$ be a \ppt no-input, single-bit output    two-party  protocol. For every $\rho>0$ there exists a \ppt $\Sim$ such that when given $(1^\kappa,t)$, $\Sim(1^\kappa,t)$ outputs two bits, $(x',y')$ such that the following holds:
Let $\REAL=\set{\REAL_\kappa}_{\kappa \in \N}$ and $\SIML=\set{\SIML_\kappa}_{\kappa \in \N}$ be ensembles defined as follows: $\REAL_\kappa=\pi(1^k)=(X_\kappa,Y_\kappa,T_\kappa)$ and let $\SIML_\kappa=(X'_\kappa,Y'_\kappa,T_\kappa)$ for $(X'_\kappa,Y'_\kappa)=\Sim(1^\kappa,T_\kappa)$. For infinitely many $\kappa \in \N$, $\REAL$ cannot be distinguished from $\SIML$ with advantage $\rho$ by \ppt algorithms.
\end{theorem}
\noindent
(A precise formal definition of computational indistinguishability with advantage $\rho$ is given in \cref{def:CompInd}. \cref{thm:intro:simulator} is formally stated in  \cref{sec:Classification} in a more general form.)

 \cref{thm:intro:simulator} comes with two caveats:

\begin{itemize}
\item The simulated  ensemble $\SIML$ is only guaranteed to resemble the real ensemble $\REAL$ on some infinite subset $\II$ of $\kappa \in \N$.
\item For $\kappa \in \II$, $\REAL$ and $\SIML$ are only weakly indistinguishable as $\rho$ is not negligible.
\end{itemize}

We do not know whether the theorem can be proven without these caveats. We mention that most of the machinery that we develop (with one notable exception) can be used towards proving a version without the caveats.
As we will demonstrate, in some cases, the caveats do not affect applications, and we can prove clean results using the theorem.

\begin{remark}[Auxiliary input simulators, and the leakage simulation lemmas]
\label{rem:leakage}
\cref{thm:intro:simulator} is similar in spirit to the so called ``leakage simulation lemma'' \cite{TTV09,JP14,VZ13,Sko16a,Sko16b,ChenChungLiao18}.

In the leakage simulation lemma one considers a pair $(T,Z)$ of random variables, and a finite class $\cal C$ of ``distinguisher functions'' (which is typically the class of circuits of some size $s$, and so we will assume this for this discussion). The lemma states that there is a ``simulator function'' $\Sim$ of circuit complexity $s'$, which on input $T$ produces a string $Z'$ such that no distinguisher $D$ from $\cal C$ can distinguish $(T,Z)$ from $(T,Z')$ with advantage greater than some parameter $\rho>0$. The complexity $s'$ is some polynomial in $s,\ell,\frac{1}{\rho}$ (where $\ell$ is the bit length of $Z$). The reader is referred to \cite{ChenChungLiao18} for a discussion of works in this framework.

There are two differences between the leakage simulation lemma and \cref{thm:intro:simulator}:
\begin{itemize}
\item The class of distinguishers $\cal C$ that we consider are randomized polynomial time machines, and we show the existence of a simulator $\Sim$ that \emph{belongs} to this class. This is crucial in our applications.
    In contrast, in the leakage simulation lemma the simulator is a circuit of size $s'>s$ and \emph{does not} belong to the class $\cal C$. Moreover, there are negative results \cite{TTV09,ChenChungLiao18} showing limitations on proving the leakage simulation lemma with $s' \le s$.
\item In \cref{thm:intro:simulator} we can only achieve $\rho>0$ that is constant, whereas the leakage simulation lemma can achieve much smaller $\rho$ (and this is crucial in some of its applications).
\end{itemize}
\end{remark}

\paragraph{Forecasters.}
In applications, it will be useful to assume that the simulators  work in the following specific fashion: there is a ``forecaster algorithm'' $\Fc$ which on input $t$, generates a description of the probability space $(X'_\kappa,Y'_\kappa)|_{T_\kappa=t}$. For technical reasons, it is helpful to think of the forecaster \Fc as a deterministic poly-time algorithm that receives its random coin $r$, as an additional input. Given input $(1^\kappa,t,r)$ the forecaster outputs three numbers:
\begin{itemize}
\item $\pa$ which is a ``forecast'' for $\pr{X_\kappa=1\mid T_\kappa=t}$.
\item $\pbz$ which is a ``forecast'' for $\pr{Y_\kappa=1\mid T_\kappa=t,X_\kappa=0}$.
\item $\pbo$ which is a ``forecast'' for $\pr{Y_\kappa=1\mid T_\kappa=t,X_\kappa=1}$.
\end{itemize}
All that is left for the simulator is to sample according to this forecast. For $p \in [0,1]$, we will use the notation $U_p$ to denote the distribution of a biased coin that is one with probability $p$. We can now restate \cref{thm:intro:simulator} in the following more general form:

\begin{theorem}[Existence of \ppt forecasters, informal]\label{thm:intro:forecaster}
Let $\pi=(\Ac,\Bc)$ be a \ppt no-input, single-bit output    two-party  protocol. For every $\rho>0$ there exists a deterministic poly-time machine \Fc that
on input $(1^\kappa,t,r)$ outputs three numbers $\pa,\pbz,\pbo \in [0,1]$ such that the following holds:
let $R_\kappa$ be a uniform  polynomially long string (intuitively $R$ serves as the random coins of \Fc), and let  $\REAL= \set{\REAL_\kappa=(\pi(1^k),R_\kappa)=(X_\kappa,Y_\kappa,T_\kappa,R_\kappa)}$ and $\SIML= \set{\SIML_\kappa=(X'_\kappa,Y'_\kappa,T_\kappa,R_\kappa)}$ be the distribution ensembles obtained by:
\begin{itemize}
\item $(\pa,\pbz,\pbo) = \Fc(1^{\kappa},T_\kappa,R_\kappa)$.

\item $X'_\kappa \from U_{\pa}$ and $Y'_\kappa \from U_{\pbb{X'_\kappa}}$.
\end{itemize}
Then for infinitely many $\kappa \in \N$, $\REAL$ cannot be distinguished from $\SIML$ with advantage $\rho$ by \ppt algorithms.
\end{theorem}
\noindent
(\cref{thm:intro:forecaster} is formally stated    in \cref{sec:Classification}.)

\cref{thm:intro:simulator,thm:intro:forecaster} may be of independent interest, and we believe that they will find more applications. This is because the simulator induces a \emph{single} distribution that is computationally indistinguishable (albeit only with advantage $\rho=o(1)$) from the real output distribution of the protocol. Moreover, in the simulated distribution $(X'_\kappa,Y'_\kappa,T_\kappa)$ (sampled using the forecaster) the variables $(X'_\kappa,Y'_\kappa)$ have \emph{information theoretic uncertainty} conditioned on $\set{T_\kappa=t}$. This enables us to use tools and techniques from information theory on the simulated distribution, and obtain results about the computational security of the original protocol (and protocols that we construct from it). Indeed, we use this approach in our applications.

We believe that a helpful analogy is the notion of \emph{computational entropy}: which given a distribution $X$ assigns a distribution $X'$ that is computationally indistinguishable from $X$ and has \emph{information theoretic uncertainty}.


\subsubsection{A Dichotomy of Single-bit Output  Two-Party  Protocols}

We now give an informal definition of uncorrelated protocols. For this purpose we introduce the following notion of a ``decorrelator''. Loosely speaking, a decorrelator is a forecaster that forecasts that $(X_\kappa,Y_\kappa)$ are independent conditioned on $T$. Once again, for technical reasons, it is helpful to think of a decorrelator as a deterministic poly-time algorithm that receives its random coin $r$, as an additional input.

\begin{definition}[$\rho$-decorrelator, and $\rho$-uncorrelated protocols, informal] \label{dfn:intro:decr}
A deterministic poly-time algorithm $\Decr(t,r)$ is a \remph{$\rho$-decorrelator} for protocol $\pi=(\Ac,\Bc)$ if the following holds: let $\REAL=\set{\REAL_\kappa}_{\kappa \in \N}$ and $\UCR=\set{\UCR_\kappa}_{\kappa \in \N}$ be ensembles defined as follows: $\REAL_\kappa=(\pi(1^k);R_\kappa)=(X_\kappa,Y_\kappa,T_\kappa,R_\kappa)$ where $R_\kappa$ is a uniformly chosen independent polynomially long string (that intuitively serves as the random coins of $\Decr$). Let $\UCR_\kappa=(X'_\kappa,Y'_\kappa,T_\kappa,R_\kappa)$ where $(\pa,\pb)=\Decr(T_\kappa,R_\kappa)$, and (independently sampled) $X'_\kappa \from U_{\pa}$ and $Y'_\kappa \from U_{\pb}$. It is required that for infinitely many $\kappa \in \N$, $\REAL$ cannot be distinguished from $\UCR$ with advantage $\rho$ by \ppt algorithms.

\noindent
A protocol $\pi$ is \remph{$\rho$-uncorrelated} if it has a $\rho$-decorrelator.
\end{definition}
\noindent (\cref{dfn:intro:decr} is formally stated in \cref{sec:Classification}.)

Loosely speaking, the fact that the randomness $R_\kappa$ appears in the two experiments, prevents the decorrelator from using $R_\kappa$ to correlate between $X'_\kappa$ and $Y'_\kappa$. In the definition the latter should appear independent, even after seeing $R_\kappa$.

We observe that $\rho$-uncorrelated protocols are uninteresting from a cryptographic viewpoint in the following sense (that is made precise in  \cref{sec:Classification}):
\begin{itemize}
\item A $\rho$-uncorrelated protocol cannot be a key-agreement protocol for $s<a+2\rho$.
\item If a ``black-box construction'' that makes $\ell$ invocations to a $\rho$-uncorrelated protocol, yields a key-agreement protocol with $s<a+3 \cdot \ell \cdot \rho$, then the black-box construction itself can be used to give a key-agreement (with the standard choices of secrecy and agreement) that does not use the original protocol. This means that a $\rho$-uncorrelated protocol cannot be converted into an ``interesting'' protocol by a black-box construction that invokes it few times.
\end{itemize}
Loosely speaking, both properties follow because an uncorrelated protocol is somewhat indistinguishable from one in which one party samples $(T_\kappa,R_\kappa)$ on his own, sends them to the other party, and each of the parties runs $\Decr(T_\kappa,R_\kappa)$ and samples its output independently (party $\Ac$ samples $X \from U_{\pa}$, and party $\Bc$ samples $Y \from U_{\pb}$).
The latter protocol can be easily attacked, and by indistinguishability, this attack also succeeds on the original protocol.

\noindent
We prove the following classification theorem:
\begin{theorem}[Dichotomy theorem, informal]\label{thm:mainInf}
Let $\pi=(\Ac,\Bc)$ be a \ppt no-input, single-bit output two-party  protocol. Then at least one of the following hold:
\begin{itemize}
\item $\pi$ can be transformed into a key-agreement protocol (for infinitely many $\kappa \in \N$).
\item For every constant $\rho>0$, $\pi$ is $\rho$-uncorrelated (for infinitely many $\kappa \in \N$).

\end{itemize}
\end{theorem}
\noindent
(\cref{thm:mainInf} is formally stated in \cref{sec:Classification}.)

The fact that we have statements on ``infinitely many $\kappa$'s'' seems to be unavoidable: it could be the case that on even $\kappa$, the protocol is a key agreement, and on odd $\kappa$, the protocol is trivial and performs no interaction.\footnote{However, the fact that we have ``for infinitely many $\kappa$'' in the two items, and not just in one, is an artifact of our proof technique, and it is natural to ask whether the result can be improved to have such a statement in only one of the items (as in the case of the Theorem of  \citet{ImpagliazzoLu89} that we mention in the next section).}

Once again, a caveat is the fact that we only get the result for $\rho=o(1)$ and not for negligible $\rho$ (as is the standard in computational indistinguishability). It is an interesting open problem to extend our results to small $\rho$.

We demonstrate the usefulness of  \cref{thm:mainInf}  below. It is important to emphasize that the caveats in \cref{thm:mainInf} (and specifically, the limitation on $\rho$) do not matter for some of our suggested applications.

\subsubsection{Perspective: Comparison to  \citeauthor{ImpagliazzoLu89} Dichotomy Theorem}
\label{sec:perspective IL}
A celebrated result of \citet{ImpagliazzoLu89} is that distributional one-way functions imply one-way functions. This can be loosely stated this way:

\begin{theorem}[\citet{ImpagliazzoLu89},  informal] Let $f$ be a poly-time computable function, then at least one of the following holds:
\begin{itemize}
\item $f$  can be transformed into a \emph{one-way function}.
\item $f$  has a \ppt \emph{inverter} (for infinitely many $\kappa \in \N$).

Namely, for every constant $c$, there exists a \ppt  $\Inv$ such that for infinitely many $\kappa \in \N$ the following holds: let $X_\kappa \from U_\kappa$ and $T_\kappa=f(X_\kappa)$. It holds that $(X_\kappa,T_\kappa)$  is $(\rho=\kappa^{-c})$-close to $(X'_\kappa,T_\kappa)$, for  $X'_\kappa=\Inv(T_\kappa)$.
\end{itemize}
\end{theorem}

This theorem is celebrated for (at least) two reasons: first, it gives a dichotomy of poly-time functions (ruling out intermediate cases). Second, it gives a methodology to show that cryptographic primitives imply one-way functions: it is sufficient to show that the primitive has a component that cannot be inverted.

 Our \cref{thm:mainInf} can be viewed as an analogous theorem for \emph{two-party protocols}: either a protocol $\pi$ implies \emph{key-agreement} or it has a \ppt \emph{decorrelator}. Indeed, \cref{thm:mainInf} gives a dichotomy of two-party protocols, and in order to show that a protocol implies key-agreement, it is now sufficient to show that it is not uncorrelated. We will present applications of this methodology in  \cref{sec:consequences}.

We remark that many of the applications of the \citet{ImpagliazzoLu89}  classification  do not require that $\rho$ is small, and would have worked just the same for constant $\rho$.\footnote{Loosely speaking, this happens whenever we have a cryptographic primitive where security can be amplified. For such protocols, a weaker version of \cite{ImpagliazzoLu89} yields that either the primitive implies one-way functions or it has a \ppt $\rho$-inverter for some constant $\rho>0$. Then, using security amplification we obtain a more secure target primitive, such that an adversary that breaks the target primitive with small success $\rho'=\kappa^{-c}$ can be transformed into one that breaks the original protocol with large success $\rho>0$.} Analogously, the fact that $\rho$ is not very small in our theorem is sometimes unimportant in applications.

\subsection{Consequences of our Dichotomy Theorem}
\label{sec:consequences}

We demonstrate the usefulness of our result by showing that it can be used to answer some open problems regarding differentially private protocols and coin flipping protocols. We now elaborate on these results.

\subsubsection{Application to Differentially Private XOR}

In  a symmetric  differentially  private  computation,  the parties wish to compute a joint function of their inputs while keeping their inputs somewhat private. This is somewhat different from the classical client-server setting that is commonly addressed in the  differentially  privacy literature, where the server, holding the data, answers the client's question while keeping the data somewhat private.

This setting is closely related to the setting of secure function evaluation: the parties $\Ac$ and $\Bc$ have private inputs $x$ and $y$, and wish to compute some functionality $f(x,y)$ without compromising the privacy of their inputs. In secure function evaluation, this intuitively means that parties do not learn any information about the other party's input, that cannot be inferred from their own inputs and outputs. This guarantee is sometimes very weak: For example, for the XOR function $f(x,y)=x \xor y$, secure function evaluation completely reveals the inputs of the parties (as a party that knows $x$ and $f(x,y)$ can infer $y$). Differentially private two-party computation aims to give some nontrivial security even in such cases (at the cost of compromising the \emph{accuracy} of the outputs).

A natural question is what assumptions are needed for such (symmetric) differentially private computation achieving certain level of  accuracy. A sequence of work showed that for certain tasks, achieving high  accuracy requires one-way functions \cite{BeimelNO08, ChanSS12,MMPRTV11,GoyalMPS2013}; some cannot even be instantiated in the random oracle model \cite{HaitnerOZ2016}; and some cannot be black-box reduced to key agreement  \cite{KhuranaMS2014}. See  \cref{sec:intro:relatedWork} for more details on these results. Recently, see more details below, \cite{Goyal2016KMPS}  have shown that a protocol for computing the XOR of \emph{optimal} accuracy (\ie that matches the client server accuracy for XOR) implies the existence of  oblivious transfer protocols (that are also sufficient for this task).

We show that  the existence of a  symmetric differential private protocol  for computing  Boolean XOR that achieves   \emph{non-trivial accuracy} (\ie better that what can be achieved when the eavesdropper is unbounded), implies the existence of a key-agreement protocol.

To prove the above result we  consider protocols in which the two parties receive inputs $x,y \in \zo$ and each outputs a bit.
A two-party protocol $\pi = (\Ac,\Bc)$ for computing the XOR functionality is \emph{$\alpha$-correct}, if

$$\pr{\pi(\Ac(x),\Bc(y)) = (x\xor y,x\xor y)} \ge \frac12 + \alpha$$

 Such a  protocol is (computationally) \emph{$\eps$-differentially private}, if for every  $x$ and efficient distinguisher $\Dc$
 $$\frac{\pr{\Dc(\view^\Ac_\pi(x,0)) = 1}}{\pr{\Dc(\view^\Ac_\pi(x,1)) = 1}} \in e^{ \pm \eps}$$
    letting $\view^\Ac_\pi(x,y)$ being  $\Ac$'s view in a random execution of $\pi(\Ac(x),\Bc(y))$;\footnote{A more general  definition allows  also an additive  error term. We address this definition in our formal theorem in  \cref{sec:DPXRtoKA}.} namely, the input of $\Bc$ remains somewhat private from the point of view of $\Ac$. And the same should hold for the privacy of $\Ac$.

The  protocol has   perfect \emph{agreement}, if the parties' output is  always the same (though might be different from the XOR). The results below are all stated \wrt such perfect agreement protocols,  though the lower bound (including ours) allows disagreement  in the magnitude of the differential privacy  parameter $\eps$.


\begin{theorem}[Differentially private XOR to key agreement, informal]\label{thm:DPXORInf} For every $\eps>0$, the existence of  $21\eps^2$-correct  $\eps$-differentially private protocol for computing  XOR, implies the existence of  an infinitely often secure key-agreement protocol.
 \end{theorem}
\noindent
(\cref{thm:DPXORInf} is formally stated in \cref{sec:DPXRtoKA}.)

The above dependency between $\eps$ and $\alpha$ is  tight since  a $\Theta(\eps^2)$-correct, $\eps$-differential private,  protocol for computing  XOR  can  be constructed (with information theoretic security) using the so-called \textit{randomized response} approach shown in \citet{W65a}. It    improves, in the $(\eps,\alpha)$ dependency aspect,  upon \citet{Goyal2016KMPS} who showed  that, for some constant $c>0$,   a  $c\eps$-correct $\eps$-differentially private XOR implies oblivious transfer, and  upon \citet{GoyalMPS2013}  who showed that  $c\eps^2$-correct $\eps$-differentially XOR implies one-way functions.

\cref{thm:DPXORInf} extends for a  weaker notion of privacy  in which differential privacy is only guaranteed  to hold against an \emph{external}  observer (assuming that the protocol's transcript explicitly states the parties common output). For such protocols, key agreement is  a sufficient assumption.\footnote{One party sends its \emph{encrypted}  input to the other party, who in turn  computes the  XOR of both inputs and publishes a noisy version (\eg flipped with probability $\frac12 - \eps$) of the outcome.} Finally, we mention that since we use \cref{thm:mainInf},  the reduction we use  to prove \cref{thm:DPXORInf}  is non black box in the adversary.


A recent subsequent work by Haitner, Mazor, Shaltiel and Silbak \cite{HMSS} improved on the above result. They showed that a non-trivial differentially private protocol for computing XOR can be used to construct a standard oblivious transfer protocol (without the infinitely often). Moreover, the dependency between $\eps$ and $\alpha$ is essentially optimal (similar to the result presented in this paper).

\subsubsection{Application to Fair Coin Flipping}

In a follow-up work,   \citet{HMO} used \cref{thm:mainInf}   to prove that key-agreement is a necessary assumption for \emph{two-party} $r$-round coin-flipping  protocol of bias smaller than $1/\sqrt{r}$ (as long as $r$ is independent of the security parameter). This partially answers a long-standing  open question  asking whether the existence of such two-party fair-coin flipping implies public-key cryptography. Previous to \citet{HMO} result, it was not even known that such protocols cannot be constructed in the random oracle model \cite{Dachman11,DachmanMM14}.

In a very high level,  \cite{HMO}  took the following approach. Assume key-agreement protocols do not exists, then the main result of this paper (\cref{thm:mainInf}) yields that any protocol, and in particular an  $r$-round coin-flipping  protocol, has a decorrelator.   \citet{HMO}  showed how to use this decorrelator to mount an efficient variant of the  \citet{CleveI93}  attack to bias the outcome of one of the parties by  $1/\sqrt{r}$. (The bound  of \cite{HMO} only holds for constant-round protocols, since for the attack to go through the  decorrelator's error  has to be smaller than  $1/\sqrt{r}$, which can only be achieved, at least using  \cref{thm:mainInf}, for constant $r$.)

\subsection{Our Technique}\label{sec:intro:Technique}

\subsubsection{A Competition of Forecasters}

In this section we explain the high level idea behind the proof of  \cref{thm:intro:forecaster}.
Our goal is to understand ``how $X_\kappa$ and $Y_\kappa$ are distributed from the point of view of a \ppt algorithm that receives $T_\kappa$ as input''. For this purpose, we set up a competition between all  \ppt forecasters.
We will use the winner in this competition as our forecaster.

Given a transcript $t$, a participant forecaster is required to output three numbers $\pa,\pbz,\pbo \in [0,1]$.
For every forecaster $\Fc$ and every $\kappa \in \N$, we associate a \emph{price} $\Price_\kappa(\Fc)$. The minimal price is obtained by a forecaster that outputs $\pa=\pr{X_\kappa=1 \mid T_\kappa=t}$ and $\pbb{b}=\pr{Y_\kappa=1 \mid T_\kappa=t,X_\kappa=b}$. Note however, that a \ppt forecaster might not be able to compute these quantities.

\paragraph{Existence of optimal forecasters.}
We will not give a precise definition of the price function in this overview. At this point, we observe that for every choice of price function where prices are in $[0,1]$, this competition has winners, in the following sense: we say that $\Fc$ is {\sf $\mu$-optimal}, if there exists an infinite subset $ \II \subseteq \N$ such that  $\Price_{\kappa}(\Fc) \le \Price_{\kappa}(\Fc') + \mu$ for every other \ppt $\Fc'$ and sufficiently large $\kappa\in \II$. This intuitively says that $\Fc$ cannot be significantly improved on the subset $\II$. We claim that for every constant $\mu>0$ there exists a $\mu$-optimal forecaster.

This follows as we can imagine the following iterative process: we start with some forecaster $\Fc$ and $\II=\N$. At each step, either $\Fc$ cannot be improved by $\mu$, on infinitely many $\kappa \in \II$ (which means that $\Fc$ is $\mu$-optimal), or else, there exists an infinite $\II' \subseteq \II$, and a forecaster $\Fc'$ that improves $\Fc$ by $\mu$ in $\II'$. In that case we set $\II=\II'$, $\Fc=\Fc'$ and continue. It is clear that at every iteration we improve the price by $\mu$, and  this can happen only $1/\mu$ times, this process shows the existence of a $\mu$-optimal forecaster.
\begin{remark}
A drawback of the argument above is that it only works for constant $\mu>0$. The distinguishing parameter $\rho$, will be selected to be say $\mu^{1/10}$, and this is why we only get the result in \cref{thm:intro:simulator}, \cref{thm:intro:forecaster} and \cref{thm:mainInf} for constant $\rho>0$. Consequently, if we could guarantee the existence of an optimal forecaster for smaller $\mu$, we will immediately improve our results. Another drawback is that this argument only works on some infinite subset $\II \subseteq \N$ and this is the reason we get ``for infinitely many $\kappa$" in our theorems. The remainder of our machinery does not require these caveats.
\end{remark}

\paragraph{Indistinguishability for optimal forecasters.}
Let $\Fc$ be a $\mu$-optimal forecaster, we can use $\Fc$ to produce a forecasted distribution (as in \cref{thm:intro:forecaster}). Namely, given $t \from T_\kappa$, we apply $\Fc(t)$ to compute $\pa(t),\pbz(t),\pbo(t)$, and use these forecasts to produce a distribution $(X'_\kappa,Y'_\kappa)$ by sampling $X'_\kappa \from U_{\pa(t)}$ and $Y'_\kappa \from U_{\pbb{X'_\kappa}(t)}$. This can indeed be done in poly-time (and in this informal discussion we omit the additional random input $r$).

We show that if a \ppt $\Dc$ distinguishes $(X_\kappa,Y_\kappa,T_\kappa)$ from $(X'_\kappa,Y'_\kappa,T_\kappa)$, then $\Dc$ can be used to construct an improved \ppt $\Fc'$ whose $\Price_\kappa(\Fc')$ is smaller than $\Price_\kappa(\Fc)$ by some function of the distinguishing advantage $\rho$.\footnote{This overall approach is also taken by some proofs of the ``leakage simulation lemma'' that was mentioned in remark \ref{rem:leakage}.} This is a contradiction to the $\mu$-optimality of $\Fc$ if $\rho$ is sufficiently large.

\smallskip
At the risk of getting too technical, let us try to explain how this argument works. The reader can skip to Section \ref{sec:technique:dichotomy} that does not depend on the next paragraph.

It is helpful to note that $(X'_\kappa,Y'_\kappa,T_\kappa)$ can be seen as $(X'_\kappa,g(X'_\kappa,T_\kappa),T_\kappa)$ where $g$ is a probabilistic function. It is helpful to consider the hybrid distribution $H=(X_\kappa,g(X_\kappa,T_\kappa),T_\kappa)$. Using a hybrid argument, we have that one of the following happens:
\begin{itemize}
\item $\Dc$ distinguishes $(X'_\kappa,g(X'_\kappa,T_\kappa),T_\kappa)$ from $H=(X_\kappa,g(X_\kappa,T_\kappa),T_\kappa)$. This induces a $\Dc'$ that distinguishes $(X'_\kappa,T_\kappa)=(U_{\pa(T_\kappa)},T_\kappa)$ from $(X_\kappa,T_\kappa)$
\item $\Dc$ distinguishes $(X_\kappa,Y_\kappa,T_\kappa)$ from $H=(X_\kappa,g(X_\kappa,T_\kappa),T_\kappa)$. This gives that there exists $b \in \zo$, and a $\Dc'$ such that $\Dc'$ distinguishes $(Y_\kappa,T_\kappa)|_{X_\kappa=b}$ from $(Y'_\kappa,T_\kappa)|_{X_\kappa=b}=(U_{\pbb{b}(T_\kappa)},T_\kappa)|_{X_\kappa=b}$.
\end{itemize}
We have made progress, in that in both cases we have reduced the number of variables from three to two, while obtaining a distinguisher $\Dc'$ that distinguishes between a ``real distribution'' and a ``forecasted distribution''. Let's assume \wlg that the first case happens. Note that $\Dc'$ obtains no distinguishing advantage on $t$ if $\Dc'(t,0)=\Dc'(t,1)$.

Assume \wlg  that $\Dc'$ is more likely to answer one on the real distribution than on the forecasted distribution. This intuitively means that on average, given a $t \from T_\kappa$, by trying out $\Dc'(t,0)$ and $\Dc'(t,1)$ we can figure out what ``$\Dc'$ thinks'' is more likely to be the bit of the forecasted distribution, and improve the forecast of \Fc. Specifically,
\begin{itemize}
\item If $\Dc'(t,0)=\Dc'(t,1)$ then $\Dc$ does not gain on $t$, and we won't modify the forecast of \Fc on $t$.
\item If $\Dc'(t,1)=1$ and $\Dc'(t,0)=0$ then ``$\Dc'$ thinks'' that \Fc's forecast for $\pr{X_\kappa=1 \mid T_\kappa=t}$ was too low, and it makes sense to increase it.
\item If $\Dc'(t,0)=1$ and $\Dc'(t,1)=0$ then ``$\Dc'$ thinks'' that \Fc's forecast for $\pr{X_\kappa=1 \mid T_\kappa=t}$ was too high, and it makes sense to decrease it.
\end{itemize}
By using this rationale, we can guarantee that the modified forecast (which can be computed in poly-time) improves upon \Fc's forecast (at least on average $t \from T_\kappa$). We choose the price function carefully, so that this translates to a significant reduction in price, contradicting \Fc's $\mu$-optimality.

\subsubsection{Using the Forecaster to Prove the Dichotomy}
\label{sec:technique:dichotomy}

In this section we explain how to prove  \cref{thm:mainInf} given  \cref{thm:intro:forecaster}. Given a protocol $\pi$, we consider the optimal forecaster \Fc from  \cref{thm:intro:forecaster} (which is \Fc from the previous section). We will once again oversimplify and ignore the random coin string $r$. Recall that on input $t \from T_\kappa$, \Fc computes three numbers $\pa,\pbz,\pbo$, and induces a forecasted distribution $(X'_\kappa,Y'_\kappa,T_\kappa)$ that is $\rho$-indistinguishable from $\pi(1^\kappa)=(X_\kappa,Y_\kappa,T_\kappa)$, and furthermore, that $\pr{X'_\kappa=1 \mid T_\kappa=t}=\pa$, and $\pr{Y'_\kappa=1 \mid T_\kappa=t,X'_\kappa=b}=\pbb{b}$.

Note that if for every possible transcript $T_\kappa$ it holds that  $\Fc(T_\kappa)$ produces $\pbz = \pbo$, then by setting $\Decr(T_\kappa)=(\pa,\pbz)$ we obtain a $\rho$-decorrelator. Increasing $\rho$ slightly, this also extends to the case where with high probability over $t \from T_\kappa$, $\pbz$ is ``not far'' from $\pbo$. If the condition above does not hold, we will want to use \Fc to convert $\pi$ into a key-agreement $\pi'$. We can use the forecaster as follows (and in fact this methodology seems quite general):
\begin{itemize}
\item When using $\pi$ as a component in $\pi'$, we can imagine that the output distribution of $\pi$ is the forecasted distribution. More precisely, we are allowed to work in the following ``information theoretic setting'': party $\Ac$ receives $X'_\kappa$, party $\Bc$ receives $Y'_\kappa$ and the adversary receives $T_\kappa$. Note that $X'_\kappa$ and $Y'_\kappa$ have \emph{information theoretic uncertainty} given $T_\kappa$, and so we can now apply techniques and protocols from the information theoretic world. Information theoretic security in the latter setup translates into computational security in the original setup (with an additive loss of $\rho$).

\item Consequently, we can use information theoretic methods to construct  key-agreement to construct $\pi'$ from the ``simulation of'' $\pi$. This then translates into computational security (with a constant loss $\rho$ in security). By using security amplification for key agreement \cite{Holenstein06b}, we can amplify this security to give key-agreement with standard choices of secrecy and agreement. (This demonstrates that the fact that $\rho$ cannot be made negligible, is not a problem, and we can get computational security with respect to negligible functions).\footnote{Continuing the analogy to computational entropy, this approach can be thought of as analogous to the constructions of  \citet{HastadImLeLu99} and following work \cite{HaitnerReVa13,VadhanZheng2012}  of pseudorandom generators from one-way functions. Indeed, a key idea in these works is that of ``computational entropy''  which given a distribution $X$ (with low real entropy) presents an indistinguishable distribution $X'$ (with a lot of entropy). This allows the construction to apply ``information theoretic tools'' (e.g., randomness extractors) on $X$ and argue that the result is pseudorandom, by imagining that the information theoretic tools are applied on $X'$. Continuing this analogy, it is often the case that ``pulling the result back'' to the computational realm, suffers a significant loss in security, and computational amplification of security is performed to obtain stronger final results.}
\item Moreover, when we work in the information theoretic setup, the honest parties are allowed to see $T$, and run the forecaster (that runs in polynomial time). This is in some sense ``non-black-box'' as the parties gain access to specific properties of the probability space $(X'_\kappa,Y'_\kappa,T_\kappa)$  by applying the forecaster on $T_\kappa$ and can use its outputs $\pa,\pbz,\pbo$ when constructing information theoretic key-agreement.
\end{itemize}

\paragraph{The one-sided von-Neumann protocol.}
The information theoretic setup described above can be thought of as follows: whenever the two parties invoke the protocol $\pi$, we can imagine that $\Ac$ receives variable $X'_\kappa$, $\Bc$ receives variable $Y'_\kappa$ and the eavesdropper receives $T_\kappa$. Moreover, $\Ac$ and $\Bc$ can use \Fc to compute all probabilities in the probability space $(X'_\kappa,Y'_\kappa)|_{T_\kappa=t}$.
We now explain how to construct a key-agreement protocol.
\begin{itemize}
\item The two parties receive $X'_\kappa$ and $Y'_\kappa$ by running $\pi$, they also receive the transcript $T_\kappa$.
\item The two parties use \Fc to compute $\Fc(T_\kappa)=(\pa,\pbz,\pbo)$. Party $\Ac$ samples an independent random variable $X''_\kappa \from U_{\pa}$ (that is, an independent variable that is distributed like $X'_\kappa$).
\item The two parties can use the von-Neumann trick \cite{vN} to obtain a shared random coin as follows: $\Ac$ informs $\Bc$ whether $X'_\kappa = X''_\kappa$.
\begin{itemize}
\item If $X'_\kappa = X''_\kappa$,  the  parties output independent uniform bits.

\item If $X'_\kappa \ne X''_\kappa$, party $\Ac$ outputs $X'_\kappa$ and party $\Bc$ outputs $Y'_\kappa$.
\end{itemize}
\end{itemize}

For every $t \in \Supp(T_\kappa)$, $\p{[X'_\kappa=1,X''_\kappa=0 \mid T_\kappa=t} = \pr{X'_\kappa=0,X''_\kappa=1 \mid T_\kappa=t}$, and consequently:
\[ \pr{X'_\kappa=1 \mid T_\kappa=t,X'_\kappa \ne X''_\kappa} = \half. \]
This means that this information theoretic key-agreement protocol has perfect secrecy. We now consider the agreement property. Recall that we are assuming that $X'_\kappa$ and $Y'_\kappa$ are correlated conditioned on some fixings of $t \from T_\kappa$. This can be used to show that the output bits of our protocol are correlated. (In the actual proof, we need a slightly more complicated protocol which also relies on $\pbz,\pbo$ to guarantee agreement, rather than just correlation).

Thus, this protocol is an information theoretic key agreement with secrecy $s=0$ and agreement $a>0$. By controlling the parameters, the gap between agreement and secrecy can be made significantly larger than $\rho$ so that we can implement our overall plan.

\subsection{Related Work}\label{sec:intro:relatedWork}

We now discuss some related work that was not yet mentioned in the previous sections.

\paragraph{Key agreement form information theoretic correlated sources.}
The question of constructing information theoretic key agreement protocols from multiple (similarly distributed) correlated triplets $(X,Y,T)$ (where Alice and Bob get $X$ and $Y$ respectively, and Eve get $T$), was posed by Maure \cite{Mau}. In the same paper, Maure also defined the secret-key rate of such triplets (sources),\footnote{Loosely speaking, the secret-key rate is the maximum rate at which Alice and Bob can agree on a secret key $S$ while keeping the rate at which Eve obtains information arbitrarily small.} and gave an upper and lower bound depending on the distribution. A better upper bound on the secret-key rate was given by Ahlswede and Csisz{\'a}r \cite{AC}, and later by Maure and Wolf \cite{MW}, using the notion of intrinsic information.     

\paragraph{Computational Key agreement.}
Computational key-agreement protocols were first introduces by Diffie and Hellman \cite{DH} assuming computational (algebraic) hardness. Dwork, Naor and Reingold \cite{DNR}, showed how to improving an imperfect public key cryptosystem to a more secure system. Holenstein \cite{Holenstein06} considered the problem of strengthening computationally secure key agreement (key agreement amplification) using hard-core sets.

\paragraph{Characterization of two-party computations.}
The most relevant result is the classification of two-party protocols in the random oracle model (ROM) given in \citet*{HaitnerOZ2016}. In this model, the parties and the adversary  are given an oracle access to a common random function,  that they can query a limited number of times. The ROM is typically used to analyze the security of cryptographic protocols  in an idealistic model, and to prove impossibility results for such protocols. In particular, an impossibility result in the ROM yields that  the security  of  protocol in consideration cannot  be based in a black-box way on one-way functions or collision resistant hash functions.

In their seminal work \citet{ImpagliazzoRu89} proved that a key-agreement protocols cannot be constructed  in the ROM. That is, they show that for any query efficient protocol (\ie polynomial query complexity) in the ROM, there exists a query efficient eavesdropper that finds the common key. \citet{HaitnerOZ2016}, using techniques developed by \citet{BarakGhidary09}, showed that for any  no-input two-party random oracle protocol there exists a query efficient mapping into a \emph{no oracle} protocol such that the distribution of the transcript and parties output are essentially the same.  Since in the non-input  setting the parties output are always uncorrelated (as far as no input protocol are concerned), the existence of such efficient mapping also tell us that interesting correlation cannot exits in the ROM. Our main result capturing the minimal assumption for (output) correlation in actual protocol (rather than the hypothetical random oracle, model)  is in a sense  the non black-box  version of the above characterization.

Other relevant results  are amplifications of weak primitives into a full-fledge ones,  and in particular that of key-agreement  \cite{Holenstein06b}  and obvious transfer  \cite{Haitner04,Wullschleger2007,CreKil88}. Such results  aims to classify the different functionalities into groups of equivalent expression power, and many of them are achieved via the study of information-theoretic two-party correlation (\aka channels):  each party, including the observer, is  given random variable from  a predetermined distribution, and their goal is to use them to achieve a cryptographic task (\ie key agreement). Our result demonstrates that going solely through the above information theoretic paradigm, is sometimes a too limited approach.

\paragraph{Minimal assumptions for differentially  private symmetric computation.}
An accuracy parameter $\alpha$ is \emph{trivial}  \wrt  a given functionality $f$ and differential privacy parameter $\eps$, if a protocol computing $f$ with such accuracy and privacy exists information theoretically (\ie with no  computational assumptions). The accuracy  parameter  is called \emph{optimal},  if it matches the bound achieved in  the client-server model.
Gaps between the trivial and optimal accuracy parameters have been shown in the multiparty case for count queries~\cite{BeimelNO08, ChanSS12}  and in the two-pary case for inner product and hamming distance functionalities~\cite{MMPRTV11}. 
\cite{HaitnerOZ2016}  showed that the same holds also when a random oracle is available to the parties,   implying that non-trivial  protocols (achieving non-trivial  accuracy)  for computing  these functionalities    cannot be black-box reduced to one-way functions.  \cite{GoyalMPS2013} initiated the study of Boolean  functions, showing  a gap  between the optimal  and trivial accuracy   for the XOR or the AND functionalities, and that non-trivial  protocols  imply  one-way functions.   \cite{KhuranaMS2014} have shown that optimal protocols for computing the XOR or AND, cannot be black-box reduced to key agreement.  Recently, \cite{Goyal2016KMPS}  have shown that optimal protocols for computing the XOR imply oblivious transfer.

 \subsection*{Paper Organization}
Standard notions and definitions are given in \cref{sec:Preliminaries}.  In \cref{sec:Classification} we formally  define simulators, forecasters,  decorrelators, and uncorrelated protocols, and state there our main results. The existence of  forecasters for every single-bit output  two-party protocol whose forecasted distribution is indistinguishable from the real one,  is proven in  \cref{sec:Forecasters}. The reduction from  correlated protocols to key agreement is proven in \cref{sec:ForecastingToKA}. Finally in \cref{sec:DPXRtoKA}, we  give the reduction from differentially private protocols  for computing  XOR  to  key-agreement protocols.

\subsection*{Acknowledgement}
We are very grateful to Omer Reingold and  Guy Rothblum for very useful discussions. We are also very grateful the anonymous referees for their detailed and helpful feedback.

\section{Preliminaries}\label{sec:Preliminaries}
\subsection{Notations}\label{sec:notations}
We use calligraphic letters to denote sets, uppercase for random variables, lowercase for values, boldface for vectors, and sans-serif (\eg \Ac) for algorithms (\ie Turing Machines). We let $\1_\cs$ denote the charectristic function of the set $\cs$.
For $n\in\N$, let $[n]=\set{1,\cdots,n}$. Let $\poly$ denote the set of all positive polynomials and let \ppt denote a probabilistic algorithm that runs in \emph{strictly} polynomial time. A function $\nu \colon \N \mapsto [0,1]$ is \textit{negligible}, denoted $\nu(\secParam) = \negl(\secParam)$, if $\nu(\secParam)<1/p(\secParam)$ for every $p\in\poly$ and large enough $\secParam$. Given an algorithm $\Dc$ getting input of the form $1^\N \times \zs$, we let $\Dc_\kappa(t) $ denote $\Dc(1^\kappa,t)$.

\paragraph{Distributions and random variables.}
For $0 \le p <1$, let $U_p$ denote the distribution of a biased coin which is one with probability $p$.  Given jointly distributed random variables $X,Y$ and $x\in \cal{X}$,  let $Y|_{X=x}$ denote the distribution of $Y$ induced by the conditioning $X=x$ (set arbitrarily if $\pr{X=x} =0$).  The \emph{statistical distance} between two random variables $X$ and $Y$ over a finite set $\Uni$, denoted $\SD(X,Y)$, is defined as $\frac12 \cdot \sum_{u\in \Uni}\size{\pr{X = u}- \pr{Y = u}}$. 

We will also use the following standard lemma to compute the statistical distance of jointly distributed random variables.
\begin{lemma}\label{lemma:SD for joint}
	Let $(X,Y)$ and $(X,Z)$ be two finite random variable with the same $X$, 
it follows that $\SD((X,Y),(X,Z))=\Exp_{x \gets X}[\SD(Y|_x, Z|_x)]$.
\end{lemma}
\begin{proof}
Set $\Uni=\sup(Y)\cup\sup(Z)$. By the definition of statistical distance, 
\begin{align*}
\SD((X,Y),(X,Z))&=\frac12 \cdot \sum_{(x,u)\in \sup(X)\times\Uni}\size{\pr{(X,Y) =(x,u)}- \pr{(X,Z) = (x,u)}}\\
				&= \sum_{x\in\sup(X)}\pr{X=x}\big(\frac12\sum_{u\in \Uni}\size{\pr{Y =u|X=x}- \pr{Z = u|X=x}}\big)\\
				&= \sum_{x\in\sup(X)}\pr{X=x}\big(\SD(Y|_x, Z|_x)\big)\\
				&=\Exp_{x \gets X}[\SD(Y|_x, Z|_x)]
\end{align*}
\end{proof}

\paragraph{Computational indistinguishability (and infinitely often variants).}

We first need the following variance of computational indistinguishability where the distinguishing advantage $\rho$ is a parameter. We also discuss infinitely often indistinguishability.

\begin{definition}[Computational indistinguishability with a parameter $\rho$]\label{def:CompInd}
For a function $\rho:\N \to \mathbb{R}$,
two distribution ensembles $X=\set{X_\kappa}_{\kappa \in \N}$, $Y=\set{Y_\kappa}_{\kappa \in \N}$ are {\sf $\rho$-indistinguishable}, denoted $X \cindist_\rho Y$, if for every \pptm $\Dc$, for every sufficiently large $\kappa \in \N$,
\[ |\Pr[\Dc(1^\kappa,X_\kappa)=1]-\Pr[\Dc(1^\kappa,Y_\kappa)=1]| \le \rho(\kappa) \]
We omit $1^\kappa$ when the security parameter $\kappa$ is clear from the context.

For an infinite set $\II \subseteq \N$, the two ensembles $X$ and $Y$ are {\sf $\rho$-indistinguishable in $\II$}, denoted $X \cindist_{\rho,\I} Y$, if the condition above holds when replacing the condition ``for every sufficiently large $\kappa \in \N$'' with ``for every sufficiently large $\kappa \in \II$''.
We say that $X$ and $Y$ are {\sf io-$\rho$-indistinguishable}, if there exists an infinite set $\II \subseteq \N$ such that $X$ and $Y$ are $\rho$-indistinguishable in $\II$.
\end{definition}


\subsection{Protocols}
Let $\pi= (\Ac,\Bc)$  be a two-party protocol. Protocol  $\pi$ is \ppt if both $\Ac$ and $\Bc$ running time is polynomial in their input length. We denote by $(\Ac(x),\Bc(y))(z)$ a random execution of $\pi$ with private inputs $x$ and $y$, and common input $z$, and sometimes abuse notation and refer to $(\Ac(x),\Bc(y))(z)$ as the parties' output in this execution.

We will mainly focus on no-input two-party single-bit output \ppt protocol: the two \ppt parties' only input is the common security parameter, given in unary, and at the end of the protocol each party output a single bit.  Throughout, we assume \wlg that  the transcript contains $1^\kappa$ as the first message.

Let $\pi= (\Ac,\Bc)$  be such two-party single-bit protocol. For $\kappa\in\N$,  let $\pi_\kappa$ be protocol $\pi$ with the common security parameter fixed (\ie hardwired) to $1^\kappa$. Protocol  $\pi$ has transcript length $m(\cdot)$, if the transcript of $\pi_\kappa$ is of length at most $m(\kappa)$. We will assume \wlg that the protocol of consideration has fixed transcript length per security parameter. For $\kappa\in \N$, let $(X^\pi_\kappa,Y^\pi_\kappa,T^\pi_\kappa)$ denote the   $\Ac$ and $\Bc$ outputs respectively, and the execution transcript, in a random execution of $\pi_\kappa$.   We sometimes denote  this triplet of random variables by  $\pi(1^\kappa)$.

\subsubsection{Key-Agreement Protocols (and infinitely often variants)}
We focus on single bit key agreement protocols.

\begin{definition}[Key-agreement protocols]\label{def:KA}
	A  \ppt single-bit output  two-party protocol $\pi = (\Ac,\Bc)$  is a secure {\sf key-agreement} \wrt a set $\I\subseteq\N$, if the following hold for $\kappa$'s in $\I$.
	\begin{description}
		
		\item[Agreement.]  $\pr{X^\pi_\kappa =  Y^\pi_\kappa} \ge 1- \negl(\kappa)$.
		
		\item[Secrecy.] For every \ppt \Ec it holds that $\pr{\Ec(T^\pi_\kappa)=X^\pi_\kappa} \le 1/2 + \negl(\kappa)$.
		
	\end{description}
\end{definition}

\begin{definition}[Key-agreement protocols]\label{def:KA}
	Let $s,a:\N \mapsto \R$ be functions. A \ppt single-bit output two-party protocol $\pi=(\Ac,\Bc)$  is an {\sf $(s,a)$-key agreement} if the following two conditions hold.
	\begin{description}
		\item[Agreement.]  $\pr{X^\pi_\kappa=Y^\pi_\kappa} \ge 1/2 + a(\kappa)$ for  sufficiently large $\kappa \in \N$.
		
		\item[Secrecy.] For every \pptm \Ec: $\pr{\Ec(T^\pi_\kappa)=X^\pi_\kappa} \le 1/2 + s(\kappa)$  for  sufficiently large $\kappa \in \N$.
		
	\end{description}
	If we omit $(s,a)$ then we mean that the key-agreement has standard choices for secrecy and agreement, namely it is a  $(\negl(\kappa),1/2-\negl(\pk))$-key agreement.
	
	Protocol $\pi$ is an {\sf $(s,a)$-key agreement in an infinite set $\II \subseteq \N$}, if the security and agreement conditions hold when replacing $\N$ above with $\II$. The protocol is an {\sf io-$(s,a)$-key agreement} if there exits an infinite set $\II \subseteq \N$ for which the protocol is an $(s,a)$-key agreement in $\II$.
\end{definition}

We make use of the following amplification result that  readily follow from \citet[Corollary 7.8.]{Holenstein06b}.
\begin{theorem}[Key-agreement amplification, \cite{Holenstein06b}]\label{thm:KeyAggAmp}
	Let $s,a\colon \N \mapsto \R$ be poly-time computable functions such  that $s(\pk)<a(\pk)^2/10$ for sufficiently large $\kappa \in \N$. Then there is a reduction converting an $(s,a)$-key agreement protocol in an infinite set $\II$ into a  (fully fledged)  key-agreement in $\II$. The reduction is fully black-box and oblivious to   $\II$.
\end{theorem}

\section{Classification of Boolean Two-Party Protocols}\label{sec:Classification}
In this section we formally define simulators, forecasters, decorelators and uncorrelated protocols discussed in \cref{sec:intro}, and formally state the main results of this paper.   Throughout this section we focus on no-input, single-bit output, two-party protocols.

\subsection{Simulators and Forecasters}\label{subsec:SimandFor}
The results of this section hold for \emph{any} no-input, single-bit output two-party protocols, even inefficient ones.

\subsubsection{Simulators}\label{subsec:Sim}
Recall that a simulator seeing the protocol transcript, outputs a pair of bits that look indistinguishable from the parties'  real outputs, from the point of view of an efficient distinguisher that sees only the protocol's transcript. We now define this concept precisely, and state our results.

\begin{definition}[Simulator]\label{def:Sim}
	A {\sf simulator } is a \ppt algorithm  that  on inputs   $(1^\kappa,t) \in 1^\ast\times \zs$  outputs two bits.
\end{definition}

We associate the following two distribution ensembles  with a  two-party protocol and a simulator.
\begin{definition}[Real and simulated distributions]\label{dfn:RealandSimDist}
	Let $\pi=  (\Ac,\Bc)$ be a  single-bit output two-party protocol, and  let  $\Sim$ be a simulator. We define the  \remph{real and simulated distribution ensembles} $\REAL^{\pi}=\set{\REAL^{\pi}_\kappa}_{\pk\in\N}$ and  $\SIML^{\pi,\Sim}=\set{\SIML^{\pi,\Sim}_\kappa}_{\pk\in\N}$ as follows. For $\kappa \in \N$, let $X_\kappa$, $Y_\kappa$ and $T_\kappa$ be the parties' outputs and protocol transcript in a random  execution of $\pi_\kappa$. Then
	\begin{description}
		\item[Real:] $\REAL^{\pi}_\kappa=(X_\kappa,Y_\kappa,T_\kappa)$.
		\item[Simulated:]  $\SIML^{\pi,\Sim}_\kappa=(\Sim_\kappa(T_\kappa), T_\kappa)$.
	\end{description}
\end{definition}
(Recall that  $\Sim_\kappa(t)$ denotes the output of $\Sim$ on input $(1^\kappa,t)$.)

The following theorem states that every single-bit output two-party protocol (even an inefficient one) has a simulator.

\begin{theorem}[Existence of simulators]\label{thm:Sim}
		For every single-bit output, two-party protocol  $\pi$, $\rho>0$ and infinite set $\I \subseteq \N$, there exist a simulator $\Sim$ and an infinite set $\I' \subseteq \I$ such that
	$$\REAL^\pi\cindist_{\rho,\I'} \SIML^{\pi,\Sim}.$$
\end{theorem}

\noindent
\cref{thm:Sim} is an immediate corollary of  the existence of forecasters  theorem given below.

\subsubsection{Forecasters}\label{subsec:forecasters}
A forecaster seeing the protocol transcript, outputs a \emph{description} of a two-bit distribution, that   looks indistinguishable from the parties'  real outputs, from the  point of view of an efficient distinguisher that sees only the protocol's transcript.   Thus, a forecaster is a specific method for constructing simulators: the resulting  simulator  outputs the two bits according to the distribution described by the forecaster.

\begin{definition}[Forecasters]\label{dfn:forecaster}
	A \remph{forecaster} $\Fc$ is a \pptm that on input $(1^\kappa,t) \in 1^\ast \times \zs$,  outputs a triplet  in $[0,1]^3$. We use $\Fc(1^\kappa,t;r)$ to denote the instantiation of $\Fc(1^\kappa,t)$ when using the string $r$ as random coins.\footnote{\label{fn:RealandForcasted} Since we only care about \ppt algorithms,  we will implicitly assume that the number of coins used by them on a given security parameter is efficiently computable.}
\end{definition}

We associate the following two distribution ensembles  with a  two-party protocol and a forecasters.   To  define these distributions, we associate triplets in  $[0,1]^3$ with distribution over  $\zo^2$ in the following way.
\begin{notation}\label{nota:tValDis}
	For  $p= (\pa,\pbz,\pbo) \in [0,1]^3$, let $U_p$ denote the random  variable over $\zo^2$ defined by  $\pr{U_p = (x,y)} = \pr{U_{\pa} = x} \cdot \pr{U_{\pbb{x}} = y}$. For $p= (\pa,\pb) \in [0,1]^2$, let $U_p$ denote the random variable $U_{(\pa,\pb,\pb)}$.
\end{notation}

With this notation, the variable $U_p=(X',Y')$ is composed of two random variables such that $\Pr[X'=1] = \pa$ and for $b \in \zo$, $\Pr[Y'=1|X'=b]=p_{\Bc|b}$. In particular, if $\pbz=\pbo$ then $(X',Y')$ are independent.

\begin{definition}[Real and forecasted distributions]\label{dfn:RealandForcasted}
	Let $\pi=  (\Ac,\Bc)$ be a  single-bit output two-party protocol  and  let  $\Fc$ be a forecaster. We define the  \remph{real and forecasted distribution ensembles} $\REAL^{\pi,\Fc}=\set{\REAL^{\pi,\Fc}_\kappa}_{\pk\in\N}$ and  $\FSC^{\pi,\Fc}=\set{\FSC^{\pi,\Fc}_\kappa}_{\pk\in\N}$ as follows.
	For $\kappa \in \N$, let $X_\kappa$, $Y_\kappa$ and $T_\kappa$ be the parties' outputs and protocol transcript in a random  execution of $\pi_\kappa$, and let 	 $R_\kappa$ be a uniform and independent string  whose length is the (maximal) number of coins used by $\Fc_\kappa$. Then,
	
\begin{description}
	\item[Real:] $\REAL^{\pi,\Fc}_\kappa=(X_\kappa,Y_\kappa,T_\kappa,R_\kappa)$.
	\item[Forecasted:]  $\FSC^{\pi,\Fc}_\kappa=(U_{p},T_\kappa,R_\kappa)$ for  $p = \Fc_\kappa(T_\kappa;R_\kappa)= (p_\Ac,p_{\Bc|0},p_{\Bc|1})$.
\end{description}
\end{definition}
(Recall that  $\Fc_\kappa(t;r)$ denotes the output of $\Fc$ on input $(1^\kappa,t)$ when using  randomness $r$.)

The computational distance between the real and  forecasted distribution measures how well the forecaster realizes  the  real distribution, in the eyes of  a computationally bounded distinguisher.

\begin{definition}[Forecaster indistinguishability]\label{def:fore:Indist}
	A  forecaster $\Fc$ is {\sf $(\rho,\I)$-indistinguishable},  for $\rho>0$ and infinite subset $\I \subseteq \N$, \wrt protocol $\pi$, if
	$$\REAL^{\pi,\Fc} \cindist_{\rho,\I} \FSC^{\pi,\Fc}.$$   	
\end{definition}
That is, for sufficiently large $\kappa\in \II$, the forecasted  and real distributions are  $\rho$ indistinguishable for poly-time distinguishers.

The following theorem states that every single-bit output two-party protocol (even inefficient one) has a forecaster.

\begin{theorem}[Existence of forecasters]\label{thm:forecaster}
	For every single-bit output two-party protocol  $\pi$, $\rho>0$ and infinite set $\I \subseteq \N$, there exist a forecaster $\Fc$ and an infinite set $\I'\subseteq \I$, such that $\Fc$ is {\sf $(\rho,\I')$-indistinguishable} \wrt $\pi$.
\end{theorem}
\cref{thm:forecaster}  is  proven in \cref{sec:Forecasters}  (appears there as \cref{thm:two:Indist}). The existence of simulators immediately follows by the above theorem.

\begin{proof}[Proof of \cref{thm:Sim}]
Let $\Fc$ be the forecaster for $\pi$ guaranteed by \cref{thm:forecaster}. Given a transcript of the protocol, the simulator runs $\Fc$ on this transcript, and outputs two bits according to the distribution described by its output.
\end{proof}

\paragraph{Correlated protocols and key agreement.}
We measure the  {correlation   of a forecaster \wrt a given distribution ensemble, as the ``conditional correlation distance'' of $\FSC^{\pi,\Fc}$. That is, the expectation over $T$, of the  statistical  distance of $\FSC^{\pi,\Fc}$ from a distribution in which the two outputs are a product.
	
We use the following notation to define the product distribution naturally induced by an arbitrary distribution over $\zo^2$.
	\begin{notation}\label{nota:ProdOfDis}
		For triplet $p= (\pa,\pbz,\pbo) \in [0,1]^3$, let $\Prod(p) = (\pa,(1-\pa) \cdot \pbz+ \pa \cdot \pbo)$.
	\end{notation}

	That is, $U_{\Prod(p)}$ is the product of  marginals distribution  of $U_p$.  We now define the  product of a forecasted distribution in the natural way.
	\begin{definition}[The product of a forecasted distribution]\label{def:two:productOfForecaster}
		For a   single-bit output two-party protocol $\pi$ and  forecaster $\Fc$, we defined the {\sf product forecasted  distribution $\PFSC^{\pi,\Fc}$ of \Fc \wrt $\pi$} by  $\PFSC^{\pi,\Fc}_\kappa=(U_{\Prod(\Fc(T_\kappa;R_\kappa))},T_\kappa,R_\kappa)$, where $T_\pk$ and $R_\pk$ are as in \cref{dfn:RealandForcasted}.
	\end{definition}
	The correlation of a forecaster \wrt a given distribution ensemble, is just the expected statistical distance between the forecasted distribution and its product.
	
	\begin{definition}[Correlated forecasters]\label{def:corelation}
		A  forecaster $\Fc$ is {\sf $(\eta,\I)$-correlated} \wrt two-party protocol $\pi$, for $\eta>0$ and $\I \subseteq \N$, if for every $\kappa \in \I$,
		$$\SD(\FSC^{\pi,\Fc}_\kappa,\PFSC^{\pi,\Fc}_\kappa)\ge \eta$$ 	
	\end{definition}

	The following  fact is immediate.
	\begin{proposition}[Indistinguishability plus low correlation implies closeness to product]\label{prop:UncorrelationPlusIndistToProd}
		Let $\pi$ be a single bit output two-party protocol and $\Fc$ be a forecaster. Assume $\Fc$ is  $(\rho,\I)$-indistinguishable \wrt $\pi$ for some $\rho>0$ and infinite set $\I\subseteq \N$ and that for $\eta>0$  there exists no infinite subset $\I' \subseteq \I$ for which $\Fc$  is $(\eta,\I')$-correlated \wrt $\pi$. Then
		$$\REAL^{\pi,\Fc} \cindist_{\rho + \eta,\I} \PFSC^{\pi,\Fc}.$$   	
	\end{proposition}

Sufficiently correlated  protocols (\ie have correlated and indistinguishable forecasters) are important since they can be used to construct key-agreement protocols.

\def\CorToKAThm
{
		Let $\pi$ be a \ppt two-party single-bit output protocol and  let $\Fc$ be a forecaster. Assume there exist  an infinite set $\I \subseteq \N$, $\rho>0$ and $\eta > 30\sqrt{\rho}$  such that $\Fc$ is $(\rho,\I)$-indistinguishable and $(\eta,\I)$-correlated \wrt $\pi$. Then there exists a  key-agreement protocol in $\I$.
}

\begin{theorem}[Key-agreement from correlated protocols]\label{thm:fulKAfromCorrelation}
	\CorToKAThm
\end{theorem}
We prove \cref{thm:fulKAfromCorrelation}   in \cref{sec:ForecastingToKA}.

\subsection{Decorrelators and the Dichotomy Theorem}\label{subsec:Dichotomy}

In the introduction we explained the concept of decorrelators and uncorrelated protocols, in informal \cref{dfn:intro:decr}. We now repeat the definition using more precise language.

\begin{definition}[Decorrelators]\label{def:decor}
A {\sf decorrelator} $\Decr$ is a \pptm that on input $(1^\kappa,t) \in 1^\ast \times \zs$,  outputs two numbers in $[0,1]$. We use $\Decr(1^\kappa,t;r)$ to denote the instantiation of $\Decr(1^\kappa,t)$ when using the string $r$ as random coins.
\end{definition}

We associate the following two distribution ensembles  with a  two-party protocol and a decorrelator.
\begin{definition}[Real and uncorrelated distributions]\label{dfn:RealandUncrDist}
	Let $\pi=  (\Ac,\Bc)$ be a  single-bit output two-party protocol, and  let  $\Decr$ be a decorrelator. We define the  \remph{real and uncorrelated distribution ensembles} $\REAL^{\pi,\Decr}=\set{\REAL^{\pi,\Decr}_\kappa}_{\pk\in\N}$ and  $\UCR^{\pi,\Decr}=\set{\UCR^{\pi,\Decr}_\kappa}_{\pk\in\N}$ as follows.
	For $\kappa \in \N$, let $X_\kappa$, $Y_\kappa$ and $T_\kappa$ be the parties' outputs and protocol transcript in a random  execution of $\pi_\kappa$, and let 	 $R_\kappa$ be a uniform and independent string  whose length is the (maximal) number of coins used by $\Decr_\kappa$ (see \cref{fn:RealandForcasted}). Then,
	\begin{description}
		\item[Real:] $\REAL^{\pi,\Decr}_\kappa=(X_\kappa,Y_\kappa,T_\kappa,R_\kappa)$.
		\item[Uncorrelated:]  $\UCR^{\pi,\Decr}_\kappa=(U_{p},T_\kappa,R_\kappa)$ for  $p = \Decr_\kappa(T_\kappa;R_\kappa)= (p_\Ac,p_{\Bc})$.
	\end{description}
\end{definition}
(Recall that  $\Decr_\kappa(t;r)$ denotes the output of $\Decr$ on input $(1^\kappa,t)$ when using  randomness $r$.)
\emph{Uncorrelated protocols}, are those protocols for which the above distributions are computational close.
\begin{definition}[Uncorrelated protocols]\label{def:UncorProtocols}
Let $\pi=  (\Ac,\Bc)$ be a  single-bit output two-party protocol, let $\rho>0$ and $\I \subseteq\N$. Decorrelator $\Decr$ is a {\sf$(\rho,\I)$-decorrelator} for $\pi$, if
\[\REAL^{\pi,\Decr}  \cindist_{\rho,\I} \UCR^{\pi,\Decr}.\]
Protocol $\pi$ is {\sf $(\rho,\I)$-uncorrelated}, if it  has a $(\rho,\I)$-decorrelator. Protocol  $\pi$ is {\sf io-$\rho$-uncorrelated}, if there exists an infinite set $\II \subseteq \N$ such that $\pi$ is $(\rho,\I)$-uncorrelated.
\end{definition}

\noindent
A few remarks are in order:
\begin{remark}[decorrelators and forecasters]
 The above definition of decorrelator can be seen as special case of forecasters, that on input $t$ (and randomness $r$) output a description of a product distribution (which can be viewed as forecaster that outputs the triplet $(p_A,p_{B|0},p_{B|1})$ where $p_{B|0}=p_{B|1}$).
\end{remark}	

\begin{remark}[The role of $R_\kappa$ in \cref{dfn:RealandUncrDist}]
We choose to include the randomness $R_\kappa$ in the experiments $\REAL^{\pi,\Decr}_\kappa$ and $\UCR^{\pi,\Decr}_\kappa$. Loosely speaking, the inclusion of $R_\kappa$ in the experiments is done to prevent a scenario where $\Decr$ uses the randomness $R_\kappa$ in order to correlate between $P_\Ac$ and $P_\Bc$. More precisely, we observe that a weaker notion (in which the experiments do not include $R_\kappa$) is not interesting (as in such a notion key-agreement protocol can be uncorrelated).

Indeed, consider a decorrelator that uses a uniform bit $R_\kappa$ and produces $P_\Ac=P_\Bc=R_\kappa$. Note that $(X'_\kappa,Y'_\kappa,T_\kappa)$ are computationally indistinguishable from a triplet $(X_\kappa,Y_\kappa,T_\kappa)$ that is the real distribution of a key-agreement. The insistence that $\Decr$ ``reveals its randomness'' $R_\kappa$ in the two experiments, prevents these problems, as can be seen formally in  \cref{thm:correlated is not ka,thm:proper} (stated in \cref{sec:properties}) which loosely say that uncorrelated protocols are not key-agreement and cannot be used to construct key-agreement.
\end{remark}

\noindent
This is the formal statement of our main theorem (that restates \cref{thm:mainInf} from \cref{sec:intro}).

\def\MainThm
{
For every \ppt single-bit output two-party protocol, one of the following holds:
\begin{itemize}
	\item For every constant $\rho>0$ and every infinite $\I\subseteq \N$, there exists an infinite set $\II' \subseteq \II$ such that the protocol is  $\rho$-uncorrelated in $\II'$.
	
	\item There exists  a  two-party io key-agreement protocol.
	
\end{itemize}

}
\begin{theorem}[Dichotomy of two-party protocols]\label{thm:main}
\MainThm
\end{theorem}
The proof of  \cref{thm:main} readily follow the observations stated  in the previous subsection.
\begin{proof}
	Let  $\rho>0$ and  let $\pi$ be a \ppt single-bit output two-party protocol.  Let $\rho' = (\rho/60)^2$. By \cref{thm:forecaster} there exists a  forecaster \Fc	that is $(\rho',\I')$-indistinguishable \wrt $\pi$, for an infinite set $\I' \subseteq \I$. If there exists an infinite subset $\I'' \subseteq \I'$ for which $\Fc$  is $(\rho/2,\I'')$-correlated \wrt $\pi$, then by \cref{thm:fulKAfromCorrelation}, there exists a two-party io key-agreement protocol.
	
	Otherwise, let $\Decr(1^\kappa,t;r) = \Prod(\Fc(1^\kappa,t;r))$ for $\Prod(p)$ being the product distribution defined by the  marginal of the distribution defined by $p$ (see \cref{nota:ProdOfDis}).  \cref{prop:UncorrelationPlusIndistToProd} yields that $\Decr$ is a $(\rho'+ \rho/2 < \rho,\I')$-decorrelator for $\pi$.
\end{proof}

\begin{remark}
Assuming io key-agreement does not exist, \cref{thm:main} says that for every $\rho>0$, the protocol is io-$\rho$-uncorrelated. We emphasize that this means that for every $\rho>0$ there exists an infinite $\II$ and a (poly-time) $(\rho,\I)$-decorrelator for $\pi$. We are guaranteed that for every $\rho>0$, the running time of the decorrelator is bounded by some unspecified polynomial, however, this polynomial might differ depending on $\rho$.  
\end{remark}

\subsubsection{Properties of Uncorrelated Protocols}\label{sec:properties}
In this section we list two properties of uncorrelated protocols (that were listed informally in the introduction). We show that:
\begin{itemize}
\item Uncorrelated protocols are not key-agreement.
\item Uncorrelated protocols cannot be transformed into key-agreement (in some precise sense described below).
\end{itemize}

\begin{theorem}[An uncorrelated protocol is not a key agreement]
\label{thm:correlated is not ka}
Let $\pi=(\Ac,\Bc)$ be a \ppt single-bit output two-party protocol. Let $\II \subseteq \N$ be an infinite set. If $\pi$ is $(\rho,\I)$-uncorrelated then for every numbers $s,a$ such that $s>a+2 \rho$, $\pi$ is not an $(s,a)$-key agreement in $\II$.
\end{theorem}

\begin{proof}
Let $\Decr$ be a  $(\rho,\I)$-decorrelator for $\pi$, let $\kappa \in \II$ and consider the distributions from \cref{dfn:RealandUncrDist}.
\begin{itemize}
\item $\REAL^{\pi,\Decr}_{\kappa}=(X_\kappa,Y_\kappa,T_\kappa,R_\kappa)$.
\item $\UCR^{\pi,\Decr}_{\kappa}=(X'_\kappa,Y'_\kappa,T_\kappa,R_\kappa)$.
\end{itemize}
We will show that there exists a \pptm $\Ec$ such that for every $\kappa \in \II$, \[ \Pr[\Ec(T_\kappa)=X'_\kappa] \ge \Pr[X'_\kappa=Y'_\kappa]. \]
As $\REAL^{\pi,\Decr}$ and $\UCR^{\pi,\Decr}$ are $\rho$-indistinguishable in $\II$, and $\Ec$ is \pptm, it follows that for every sufficiently large $\kappa \in \II$:
\[ \Pr[\Ec(T_\kappa)=X_\kappa] \ge \Pr[X_\kappa=Y_\kappa] - 2\rho. \]
Which gives the required consequence that $\pi$ is not an io-key-agreement with a gap larger than $2\rho$ between agreement and secrecy.

We now define the \pptm $\Ec$. Given input $t \in \Supp(T_\kappa)$, $\Ec$ samples a uniform string $r$ and applies $\Decr_\kappa(t;r)=(\pa,\pb)$.
It then outputs ``one'' iff $\pa \ge \half$. Note that for fixed $(t,r)$,
\[ \Pr[\Ec(T_\kappa)=X'_\kappa|T_\kappa=t,R_\kappa=r] = \max(\pa,1-\pa). \]
On the other hand, note that:
\[ \Pr[X'_\kappa=Y'_\kappa|T_\kappa=t,R_\kappa=r] = \pa \cdot \pb + (1-\pa) \cdot (1-\pb) \le \max(\pa,1-\pa). \]
By averaging, we conclude that:
\[ \Pr[\Ec(T_\kappa)=X'_\kappa] \ge \Pr[X'_\kappa=Y'_\kappa] \]
and the theorem follows.
\end{proof}

We want to show that uncorrelated protocols cannot be ``transformed'' into key-agreement.
We will be interested in a scenario in which a ``black-box'' transformation invokes a $\rho$-uncorrelated protocol $\pi=(\Ac,\Bc)$, $\ell$ times, in order to construct a target protocol
$\Po=(\Ao,\Bo)$. We can consider three types of transformations (in increasing order of strength)
\begin{itemize}
\item Transformations in which $\Po$ only requires the outputs of the invocations (we call these black-box).
\item Transformations in which in addition to the outputs, also use the transcripts of the $\ell$ invocations (we call these proper).
\item Transformations that in addition to the outputs, and transcripts also use the parties' views of the $\ell$ invocations (we call these general).
\end{itemize}
We will give a precise definition shortly.

Note that in an uncorrelated protocol, it could be the case that there is a ``hidden key-agreement'' where following the protocol, the views of the two parties allow them to agree on a secret key. For example, the parties may run a key-agreement protocol but decide that their ``formal outputs'' $X,Y$ are constants, and keep the key hidden in their view. Therefore, we cannot expect to show limitations on general transformations.

We will be able to show limitations on proper transformations that transform $\rho$-uncorrelated protocols into key-agreement protocols. We will explain below that the transformation that constructs an io key-agreement from the original protocol in \cref{thm:main}, is a proper transformation.

Our limitations will be of the form: If a proper transformation constructs key-agreement from some protocol, then one can construct key-agreement \emph{without} using the original protocol.

The argument for the limitation works by simply noting that \ppt parties cannot distinguish the real output distribution of $\pi$ from a simulated distribution of $\pi$, and so we can replace the $\ell$ real executions by uninteresting $\ell$ simulations, and still obtain a key-agreement protocol (with reduced gap between agreement and secrecy by a factor of $O(\ell \cdot \rho)$). Thus, if some gap remains, we can construct a meaningful key-agreement without using the original protocol.

We now state this result formally. We start by formally defining proper transformations.

\begin{definition}[Proper transformation]\label{def:proper}
Let $\pi=(\Ac,\Bc)$ be a \ppt  single-bit output  two-party protocol. We say that a protocol $\Po$ is constructed from $\pi=(\Ac,\Bc)$ using a {\sf proper transformation} in $\ell$ invocations, if it has the following form.

\begin{protocol}[$\Po=(\Ao,\Bo)$]
	\item Input: Security parameter $1^\kappa$.
	
	\item Operation:	
\begin{enumerate}

\item The parties $\Ao$ and  $\Bo$ engage in $\ell$ invocations of $(\Ac,\Bc)(1^\kappa)$, where $\Ao$ plays the role of $\Ac$, and $\Bo$ play the role of $\Bc$.

Let $x=(x^1,\ldots,x^{\ell})$, $y=(y^1,\ldots,y^{\ell})$ and $t=(t^1,\ldots,t^{\ell})$, denote the parties outputs and transcripts in the above executions.

\item The  parties $\Ao$ and $\Bo$ engage in a random execution of  $(\Ah(x),\Bh(y))(1^\kappa,t)$,  where $\Ph= (\Ah,\Bh)$ is an {\sf arbitrary} \ppt protocol, $\Ao$ plays the role of $\Ah$, and $\Bo$ play the role of $\Bh$.

  The parties output their outputs in the above execution.

\end{enumerate}
\end{protocol}
\end{definition}

The following theorem shows that if $\pi$ is $\rho$-uncorrelated, and if it is used by a proper transformation to construct a key agreement protocol in $\ell<\frac{1}{3\rho}$ invocations, then it is possible to take the proper transformation, and use it to construct a weak key-agreement protocol (without using the original protocol). This can be interpreted as saying that in fact, it was the transformation that constructed the key agreement, and the original protocol is uninteresting.

\begin{theorem}\label{thm:proper}
Let $\II \subseteq \N$ be an infinite set and let $\pi$ be a \ppt single-bit output two-party protocol that is $(\rho,\I)$-uncorrelated. Let $\Po$ be a \ppt single-bit output  two-party protocol that is constructed from $\pi$ using a proper transformation in $\ell$ invocations. If $\Po$ is an $(s,a)$-key agreement in $\II$, then the following protocol  is  an $(s+\ell \cdot \rho,a-\ell \cdot \rho)$-key agreement in $\II$.

\begin{protocol}[$\Pt= (\At,\Bt)$]
	
	\item Input: Security parameter $1^\kappa$.
	
	\item Operation:
	\begin{enumerate}
\item $\At$ samples   $\ell$  pairs $(t^{i},r^i)$ from $(T^{i}_\kappa,R^i_\kappa)$:  for every $i \in [\ell]$ it  independently emulates $\pi(1^\kappa)$  (on its own), sets  $t^i$ to be the emulation transcript, and tosses an independent $r^i$.

\item $\At$ sends $(t^1,r^1),\ldots,(t^{\ell},r^{\ell})$ to $\Bt$.

\item For $i \in [\ell]$:

	\begin{enumerate}
\item Both $\At$ and $\Bt$ invoke   $\Decr_\kappa(t^i,r^i)$, for $\Decr$ being the guaranteed decorrelator for $\pi$, and let $(p^i_{\Ac},p^i_{\Bc})$ be the outputs. 

\item Party $\Ac$ samples $(x')^i \from U_{p^i_{\Ac}}$ and $\Bc$ samples $(y')^i \from U_{p^i_{\Bc}}$.

(I.e., the two-parties perform the simulated experiment in the $i$'th coordinate.)

\end{enumerate}
Let $x'=((x')^1,\ldots,(x')^\ell)$ and $y'=((y')^1,\ldots,(y')^\ell)$

\item The parties $\At$ and $\Bt$ engage in $(\Ah(x'),\Bh(y'))(1^\kappa,t)$, where $\Ph= (\Ah,\Bh)$ is the (arbitrary) protocol used in the definition of  $\Po$, $\At$ plays the role of $\Ah$, and $\Bt$ play the role of $\Bh$.

The parties  output their outputs in the above execution.
\end{enumerate}
\end{protocol}
\end{theorem}

\begin{proof}
	Let $\Xo_\pk$, $\Yo_\pk$ and $\To_\pk$ be the parties outputs and protocol transcript,  in a random execution of $\Po(1^\pk)$, and let $\Zo_\kappa = (\Xo_\pk, \Yo_\pk,\To_\pk)$. Since  $\Po$ is an $(s,a)$-key-agreement in $\II$, for every sufficiently large $\kappa \in \II$,
\begin{align*}
\Pr[\Xo_\kappa=\Yo_\kappa] \ge \half + a(\kappa)
\end{align*}
and for every \ppt $\Ec$, for every sufficiently large $\kappa \in \II$
\begin{align*}
\Pr[\Ec(\To_\kappa)=\Xo_\kappa] \le \half + s(\kappa)
\end{align*}

Let $\Xt_\pk$, $\Yt_\pk$ and $\Tt_\pk$ be the parties outputs and protocol transcript,  in a random execution of $\Pt(1^\pk)$, and let $\Zt_\kappa = (\Xt_\pk, \Yt_\pk,\Tt_\pk)$. We will argue that $\set{\Zo_\kappa}_{\kappa\in\N}$ and  $\set{\Zt_\kappa}_{\kappa\in\N}$ are $(\ell \cdot \rho)$-indistinguishable in $\II$, meaning that the two inequalities above also hold (with an additive ``error factor'' of $\ell \cdot \rho$) when replacing $\Zo_\kappa$ with $\Zt_\kappa$. This will mean that $\Pt$ is a $(s+\ell \cdot \rho,a-\ell\cdot \rho)$-key agreement in $\II$.

Indeed, note that the only difference between $\Po$ and $\Pt$ is that in $\Po$ the parties use $\ell$ invocations of the real experiment whereas in $\Pt$ they use $\ell$-invocations of the simulated experiment. By the hybrid argument and the fact that all protocols are \ppt, it indeed follows that the distribution ensembles of $\Po$ and $\Pt$ are $(\ell \cdot \rho)$-indistinguishable in $\II$, as required.
\end{proof}

We remark that the io key-agreement achieved in \cref{thm:main} works by using a proper transformation that invokes the original protocol $\ell$ times (where $\ell$ is a constant) in order to construct an io-$(s,a)$-key-agreement with constant $s<a$ (that protocol is later amplified into an io key-agreement with the standard choices of agreement and secrecy). By \cref{thm:proper}, if the original protocol is $\rho$-uncorrelated for every $\rho>0$, then the existence of such a transformation implies key-agreement (without relying on the original protocol).




\section{Existence of Forecasters}\label{sec:Forecasters}
In this section we  prove  \cref{thm:forecaster}, that guarantees the existence of a forecaster for any single-bit output two-party protocol. Recall that a forecaster seeing the  protocol transcripts, outputs a description of the distribution that   aims to be indistinguishable from the parties' output, given this transcript.

We start, \cref{sec:OneSidedForcasters}, by considering   the one-sided variant of such a creature that  we call \textit{one-sided forecasters}. Such one-sided forecasters try to describe the output of \emph{one}  of the parties, possibly when conditioning on the other party output.  In \cref{sec:TwoSidedForcasters} we  use the machinery developed in \cref{sec:OneSidedForcasters} for showing the existence of an indistinguishable  forecaster for the distribution of both parties. To make distinction between the one-sided and two-sided case clear,  in this section we call the latter  \textit{two-sided forecasters}.

Rather than considering the distributions induced by   protocols,  we   consider the more general settings of arbitrary distribution ensembles.

\subsection{One-Sided Forecasters}\label{sec:OneSidedForcasters}
Given a distribution $Z=(V,T)$ over $\zo \times \zs$, we are interested in how well  an efficient  algorithm forecasts the probability space $V|_{T=t}$ when given $t$ as input. We call such an algorithm a one-sided forecaster.

\begin{definition}[One-sided forecasters]\label{dfn:one:forcaster}
A \remph{one-sided forecaster}  is a \ppt algorithm that on input  pair $(1^\kappa,t) \in 1^\ast\times \zs$, outputs a number in $[0,1]$.
\end{definition}
Recall, that we use the abbreviation  $\Fc_\kappa(\cdot)=\Fc(1^\kappa,\cdot)$.

\paragraph{Real and forecasted distributions.}
We associate the following two distribution ensembles, with a one-sided forecaster  and a distribution ensemble over  $\zo \times \zs$.

\begin{definition}[Real and forecasted distributions]\label{dfn:one:RealandForcasted}
	For a one-sided forecaster $\Fc$ and an ensemble of finite distributions  $Z = \set{Z_\kappa=(V_\kappa,T_\kappa)}_{\pk\in\N}$ over $\zo \times \zs$, we define the  real and forecasted distributions  $\REAL^{Z,\Fc}$ and  $\FSC^{Z,\Fc}$ by
	\begin{description}
		\item[Real:] $\REAL^{Z,\Fc}_\kappa=(V_\kappa,T_\kappa,R_\kappa)$.
		\item[Forecasted:]  $\FSC^{Z,\Fc}_\kappa=(U_{\Fc(T_\kappa;R_\kappa)},T_\kappa,R_\kappa)$.
	\end{description}
	Where $R_\kappa$ is a uniform and independent string  whose length is the (maximal) number of coins used by $\Fc_\kappa$,\footnote{Since we only care about \pptm forecasters,  we  implicitly assume that the number of coins used by the forecaster on  ($1^\kappa,t\in \Supp(T_\kappa))$ is efficiently computable.} $\Fc_\kappa(t;r)$ denotes the output of $\Fc_\kappa$ on input $t$ and randomness $r$, and $U_p$ stand for the Boolean  random variable taking the value one with probability $p$.
\end{definition}
Namely, $\REAL^{Z,\Fc}$ is just $Z$ concatenated with the randomness of the length used by $\Fc$, where $\FSC^{Z,\Fc}$ is the distribution forecasted by $\Fc$ (given $T$ and $R$ as input).

\paragraph{Indistinguishability.}
The computational distance between the real and  forecasted distribution measures how well the forecaster realizes  the  real distribution, in the eyes of  a computationally bounded distinguisher.

\begin{definition}[Forecaster indistinguishability]\label{sef:one:Indist}
A one-sided forecaster $\Fc$ is {\sf $(\rho,\I)$-indistinguishable} \wrt an ensemble of finite distributions  $Z=\set{Z_\kappa=(V_\kappa,T_\kappa)}_{\pk\in\N}$ over $\zo \times \zs$, for $\rho>0$ and $\I \subseteq \N$, if
$$\REAL^{Z,\Fc} \cindist_{\rho,\I} \FSC^{Z,\Fc}.$$   	
\end{definition}
That is, for every sufficiently large  $\kappa\in \I$, the forecasted  and real distributions are  $\rho$ indistinguishable for efficient  distinguishers. The following is our main result for one-sided forecasters.
\begin{theorem}[Existence of indistinguishable one-sided forecaster]\label{thm:one:Indist}
For every ensemble of finite distributions  $Z = \set{Z_\kappa=(V_\kappa,T_\kappa)}_{\pk\in\N}$ over $\zo \times \zs$, $\rho>0$ and infinite $\I\subseteq\N$, there exists a one-sided forecaster $\Fc$ and an infinite set $\I' \subseteq \I$, such that $\Fc$ is {\sf $(\rho,\I')$-indistinguishable} for $Z$.
\end{theorem}
The proof of  \cref{thm:one:Indist}  readily follow from its two-sided equivalent proven in the next section.

\paragraph{Price of one-sided forecasters.}
We associate a price function  with  a given ensemble of finite distributions of the above form and a one-sided forecasters. The function  intuitively measures the quality of  the forecaster (a smaller price corresponds to a better forecast).

\begin{definition}[Price of a one-sided forecasters]\label{dfn:one-sidedPrice}
Given a \ppt algorithm $\Fc$ that on input pair $(1^\kappa,t) \in 1^\ast\times \zs$ outputs a number in $\R$, and an ensemble of finite distributions $Z = \set{Z_\kappa=(V_\kappa,T_\kappa)}_{\pk\in\N}$ over $\zo \times \zs$. For every $\kappa \in \N$, we define the \remph{price of $\Fc_\kappa$ \wrt $Z_\kappa$}  by
\[\Price_{Z_\kappa}(\Fc_\kappa)= \ex{(\Fc_\pk(T_\kappa)-V_\kappa)^2}\]
where the expectation is taken over the distribution $Z_{\kappa}$ and the random coins of $\Fc_\pk$.
\end{definition}

Note that the price function is set up so that the minimal price is achieved by a one-sided forecaster $\Fc$ that on input $t \in \Supp(T_\kappa)$ outputs $q_t=\pr{V_\kappa=1 \mid T_\kappa=t}$. This is stated formally in the following claim:

\begin{claim}\label{dfn:minPrice}
Given a function $\Fc$ and a finite distributions $Z=(V,T)$ over $\zo \times \zs$. For every $t\in \Supp(T)$, it holds that $\ex{(\Fc(t)-V)^2| T=t}$ is minimal when $\Fc(t)=\pr{V=1 \mid T=t}$.
\end{claim}
\begin{proof}
	Let $q_t=\Fc(t)$, compute 
	\begin{align*}
		\ex{(\Fc(t)-V)^2| T=t}&=\ex{(V-q_t)^2| T=t}\\
							  &=\pr{V=1\mid T=t}\cdot(1-q_t)^2+\pr{V=0\mid T=t}\cdot q_t^2\\
							  &=\pr{V=1\mid T=t}\cdot(1-2q_t+q_t^2)+\pr{V=0\mid T=t}\cdot q_t^2\\
							  &=q_t^2 -2q_t\cdot\pr{V=1\mid T=t}+\pr{V=1\mid T=t}
	\end{align*}
	The above is a quadratic function, thus by deriving the above expression we get that the minimal value is obtained for $q_t=\pr{V=1 \mid T=t}$.
	\end{proof}

A key observations about one-sided forecasters is the connection between distinguishability and price improvement proven in the next section.

\subsubsection{Distinguishability to  Price Improvement}\label{sec:one:DistinguisherGivesImprovedForecaster}
Our main technical lemma for one-sided forecaster is that a distinguisher for such a forecaster can be used to get a forecaster with an improved price value.

\begin{lemma}[Distinguishability imply improved forecaster]\label{lemma:DistinguisherImplyImprove}
	Let $\Fc$ be a  one-sided forecaster and let $Z = \set{Z_\kappa=(V_\kappa,T_\kappa)}_{\pk\in\N}$ over $\zo \times \zs$ be  an ensemble of finite distributions. Assume there exists a \pptm \Dc
	and  an infinite $\II\subseteq \N$, such that for every $\pk\in\II$,
	\begin{align}\label{eq:DistinguisherImplyImprove}
	\size{\pr{\Dc_\pk(\REAL^{Z,\Fc}_\pk)=1} - \pr{\Dc_\pk(\FSC^{Z,\Fc}_\pk)=1 }}  > \rho
	\end{align}
	Then there exists a forecaster $\Fc'$ and an infinite subset $\II'\subseteq \II$, such that for every $\pk\in\II'$,
	$$\Price_{Z_\pk}(\Fc_\pk)-\Price_{Z_\pk}(\Fc'_\pk)>\rho^2.$$
\end{lemma}

\begin{proof}
	Assume there exists  \pptm $\Dc$ and  infinite $\II$ for which \cref{eq:DistinguisherImplyImprove} holds for every $\kappa \in \I$.
 Let $m_\kappa$ be a bound on the number of coins used by $\Dc_\kappa$ on inputs drawn from $\REAL^{Z,\Fc}_\pk$ or $\FSC^{Z,\Fc}_\pk$, and let $R^\Dc$ be an independent uniform string of length $m_\kappa$. We assume \wlg that for an infinite subset $\II'\subseteq\II$, for every $\pk\in\II'$ it holds that
	
	$$\pr{\Dc_\pk(\REAL^{Z,\Fc}_\pk;R^\Dc_\pk)=1} - \pr{\Dc_\pk(\FSC^{Z,\Fc}_\pk;R^\Dc_\pk)=1 }> \rho.$$
	
	The following  algorithm  uses $\Dc$ for  finding the subset of inputs  to be changed for  getting a better  forecast.
	
	 \begin{algorithm}[$\Fh^{\Fc,\Dc}_\gamma$]\label{alg:DistinguisherImplyImprove}
	 	\item Parameters: $\gamma>0$.
	 	\item Oracles:  algorithms $\Dc$ and $\Fc$.
	 	\item Input:  $(1^\kappa,t;r,r^\Dc)$. 
	 	(Comment: Here $r^\Dc$ denotes the randomness used by $\Dc$.)
	 	\item Operation:
	 	\begin{enumerate}
	 		\item If $\Dc_\pk(1,t,r;r^\Dc)=1$ and $\Dc_\pk(0,t,r;r^\Dc)=0$, output $\Fc_\pk(t;r)+\gamma$.
	 		\item If $\Dc_\pk(0,t,r;r^\Dc)=1$ and $\Dc_\pk(1,t,r;r^\Dc)=0$, output $\Fc_\pk(t;r)-\gamma$.	
	 		\item Else, output $\Fc_\pk(t;r)$.
	 	\end{enumerate}
Note that the output might not belong to $[0,1]$.	
\end{algorithm}
	 Since  $\Dc$ and $\Fc$ are \pptm, so is $\Fh^{\Fc,\Dc}_\gamma$. The following claim states that for the right choice of $\gamma$, the above algorithm yields an improved forecasters.
	 \begin{claim}\label{claim:DistinguisherImplyImprove}
	 	 Let $\gamma\in [0,\rho]$ and  let $\Fh=\Fh^{\Fc,\Dc}_\gamma$ be according to  \cref{alg:DistinguisherImplyImprove}. Then $\Price_{Z_\pk}(\Fc_\pk)-\Price_{Z_\pk}(\Fh_\pk)>\gamma(2\rho-\gamma)$ for every $\pk\in \II'$.
	 \end{claim}
	 The proof of \cref{claim:DistinguisherImplyImprove} is given below,  but we first use it to conclude the proof of the lemma. By taking $\gamma=\rho$,  \cref{claim:DistinguisherImplyImprove} yields the desirable result that  $\Price_{Z_\pk}(\Fc_\pk)-\Price_{Z_\pk}(\Fh_\pk)>\rho^2$. Still, algorithm $\Fh$ may not be a valid forecaster, since it may output values outside of $[0,1]$. Fortunately, this is not an issue, since we can use it to define the following valid forecaster $\Fc'$ that preforms as well as $\Fh$. For  $\pk\in \N$, define
	
		\[\Fc'_\pk(t,r,r^\D)= \left\{ \begin{array}{ll}
					\Fh_\pk(t,r,r^\D) & \mbox{if } \Fh_\pk(t,r,r^\D)\in [0,1] \\
									\ 1 & \mbox{if } \Fh_\pk(t,r,r^\D)>1\\
							        \ 0 & \mbox{if } \Fh_\pk(t,r,r^\D)<0
		\end{array} \right. \]
We claim that $\Price_{Z_\pk}(\Fc'_\pk) \le \Price_{Z_\pk}(\Fh_\pk)$. This follows by \cref{dfn:minPrice} since the price function is defined such that the term $|\Fc(T_\kappa)-V_\kappa|$ does not increase by making sure that the number forecasted by $\Fc$, is in $[0,1]$, as done above. It follows that $\Price_{Z_\pk}(\Fc'_\pk)\le \Price_{Z_\pk}(\Fh_\pk)<\Price_{Z_\pk}(\Fc_\pk)-\rho^2$, for every $\pk\in \II'$, concluding the proof.
\end{proof}

\newcommand{\Goz}{{\cG_{10}}}
\newcommand{\Gzo}{{\cG_{01}}}
\newcommand{\oG}{{\overline{\cG}}}
\paragraph{Proof of  \cref{claim:DistinguisherImplyImprove}.}
\begin{proof}[Proof of \cref{claim:DistinguisherImplyImprove}]
Fix $\pk\in \II'$ and omit it when clear from the context. Let $T'=(T,R,R^\Dc)$ and  for $t'=(t,r,r^\Dc)\in \Supp(T')$ let $\Fc(t')=\Fc(t;r)$.  Let
	$\Goz=\set{t' \colon\Dc(0,t')=1 \land \Dc(1,t')=0}$,  $\Gzo=\set{ \Dc(0,t')=0 \land \Dc(1,t')=1}$, $\cG=\Gzo\cup\Goz$, and let $\oG=\Supp(T')\setminus \cG$. For a given set $\cS$ let $\1_\cS(\cdot)$ denote the characteristic function of the set, that is, $\1_\cS(t')=1$ if $t'\in\cS$ and $\1_\cS(t')=0$ otherwise.

	We make the following observations (proven below) regarding  the above sets
		\begin{claim}\label{claim:DistinguisherForecasterImp}
		The followings hold.
			\begin{itemize}
				\item	
				
				$\ex{(\Fc(T')-V)^2 \cdot  \1_\oG(T')} - \ex{(\Fh(T')-V)^2\cdot \1_\oG(T')}=0$.
				
				\item	
				
				$\ex{(\Fc(T')-V)^2\cdot  \1_\Gzo(T')}-\ex{(\Fh(T')-V)^2 \cdot \1_\Gzo(T')} =-2\gamma\cdot\ex{(\Fc(T')- V) \cdot  \1_\Gzo(T')}-\gamma^2 \cdot \pr{T' \in \Gzo}$.
				
				\item
				
				$\ex{(\Fc(T')-V)^2  \cdot \1_\Goz(T')}-\ex{(\Fh(T')-V)^2\cdot \1_\Goz(T')} =2\gamma\cdot\ex{(\Fc(T')- V) \cdot  \1_\Goz(T')}-\gamma^2 \cdot \pr{T' \in \Goz}$.
			\end{itemize}
		\end{claim}

\begin{claim}\label{claim:DistinguisherValue}
	$\ex{(\Fc(T') - V) \cdot \1_{\Goz}(T')} -\ex{(\Fc(T')-V) \cdot \1_{\Gzo}(T')} >\rho.$
\end{claim}
Given the above claims, we deduce that
		\begin{align*}
		\lefteqn{\Price(\Fc)-\Price(\Fh)}\\
		&=  \ex{(\Fc(T')-V)^2} - \ex{(\Fh(T')-V)^2} \nonumber\\
		&=  \ex{(\Fc(T')-V)^2 \cdot  \1_\oG(T')} +\ex{(\Fc(T')-V)^2 \cdot  \1_\Gzo(T')} + \ex{(\Fc(T')-V)^2 \cdot  \1_\Goz(T')} \nonumber\\
		&\quad  -   \ex{(\Fh(T')-V)^2 \cdot  \1_\oG(T')} -\ex{(\Fh(T')-V)^2 \cdot  \1_\Gzo(T')} - \ex{(\Fh(T')-V)^2 \cdot  \1_\Goz(T')} \nonumber\\
		&= 2\gamma(\ex{(\Fc(T')- V) \cdot  \1_\Goz(T')} - \ex{(\Fc(T')- V) \cdot  \1_\Gzo(T')}) - \gamma^2 \cdot \pr{T' \in \cG}. \nonumber\\
		&\ge 2\gamma \rho - \gamma^2\nonumber\\
		&> \gamma(2\rho - \gamma).\nonumber
		\end{align*}
		The third equality is by \cref{claim:DistinguisherForecasterImp}, and the inequality is by \cref{claim:DistinguisherValue}.
	\end{proof}

	\begin{proof}[Proof of \cref{claim:DistinguisherForecasterImp}]
The first item  holds since by definition  $\Fh(t')=\Fc(t')$ for every $t'\notin \cG$. For the second item, since $\Fh(t')=\Fc(t')+\gamma$ for every $t'\in \Gzo$, it holds that
 \begin{align*}
		\lefteqn{\ex{(\Fh(T')-V)^2\cdot  \1_\Gzo(T')}}\\
		&=\ex{(\Fc(T')+\gamma -V)^2\cdot  \1_\Gzo(T')}\\
		&=\ex{(\Fc(T')-V)^2\cdot  \1_\Gzo(T')} +2\gamma\cdot\ex{(\Fc(T')-V) \cdot  \1_\Gzo(T')}+\gamma^2 \cdot \pr{T' \in \Gzo}.
 \end{align*}
 For the third item, a similar  calculation yields that
 \begin{align*}
	\lefteqn{\ex{(\Fh(T')-V)^2\cdot  \1_\Goz(T')}}\\
	&= \ex{(\Fc(T')-V)^2\cdot  \1_\Goz(T')} - 2\gamma\cdot\ex{(\Fc(T')-V) \cdot  \1_\Goz(T')}+\gamma^2 \cdot \pr{T' \in \Goz}.
 \end{align*}
\end{proof}

\begin{proof}[Proof of \cref{claim:DistinguisherValue}]
	Since $\Dc(1,t')=\Dc(0,t')$ for every $t'\notin\cG$, it holds that $\ex{\Dc(V,T') \cdot \1_{\oG}(T')} = \ex{\Dc(U_{\Fc(T')},T') \cdot  \1_{\oG}(T')}$.	Since, by assumption, $\ex{\Dc(V,T')} - \ex{\Dc(U_{\Fc(T')},T')}  > \rho$, we conclude that
	\begin{align*}
	\ex{\Dc(V,T')\cdot \1_\cG(T')} - \ex{\Dc(U_{\Fc(T')},T') \cdot \1_\cG(T')}  >\rho
	\end{align*}
	By  definition of $\Goz$,

	\begin{align*}
\lefteqn{ \ex{\Dc(V,T') \cdot \1_\Goz(T')} - \ex{\Dc(U_{\Fc(T')},T')  \cdot \1_\Goz(T')}}\\
	 &=\pr{(\Dc(V,T') \cdot \1_\Goz(T'))=1} -\pr{(\Dc(U_{\Fc(T')},T')  \cdot \1_\Goz(T'))=1} \nonumber\\
	 &=\pr{\1_\Goz(T')=1}-\ex{V \cdot \1_\Goz(T')}-(\pr{\1_\Goz(T')=1}-\ex{U_{\Fc(T')} \cdot \1_\Goz(T')})\nonumber\\
	&=\ex{(\Fc(T')-V )\cdot \1_\Goz(T')},\nonumber
	\end{align*}
	and similarly
	\begin{align*}
	\ex{\Dc(V,T') \cdot \1_\Gzo(T')} - \ex{\Dc(U_{\Fc(T')},T')  \cdot \1_\Gzo(T')}=-\ex{(\Fc(T')-V )\cdot \1_\Gzo(T')}
	\end{align*}
	
	Since $\Gzo$ and $\Goz$ are a partition of $\cG$, we conclude that    $\ex{(\Fc(T')-V )\cdot \1_\Goz(T')}-    \ex{(\Fc(T')-V )\cdot \1_\Gzo(T')}  >\rho$.
\end{proof}

\subsection{Two-Sided Forecasters}\label{sec:TwoSidedForcasters}
Given a distribution $Z=(X,Y,T)$  over $\zo^2 \times \zs$, we  are interested in how well an efficient algorithm  forecasts the probability space $(X,Y)|_ {T=t}$, given $t$ as input. We call such an algorithm a two-sided forecaster. Since the probability space $(X,Y)|_ {T=t}$ is  determined by three quantities:

\begin{itemize}
\item $\Pr[X=1 \mid T=t]$,
\item $\Pr[Y=1 \mid T=t,X=0]$ and
\item $\Pr[Y=1 \mid T=t,X=1]$,
\end{itemize}
A two-sided forecaster  $\Fc$ should output a triplet of numbers $(p_1,p_2,p_3)\in [0,1]^3$.

\begin{definition}[Two-sided forecasters]\label{dfn:two:forecaster}
A \remph{two-sided forecaster} $\Fc$ is a \pptm that on input $(1^\kappa,t) \in 1^\ast \times \zs$,  outputs a triplet  in $[0,1]^3$.
\end{definition}

\paragraph{Real and forecasted distributions.}
Similarly to the one-sided case, we associate the following two distribution ensembles  with  a given ensemble of finite distributions (of the right  form) and a two-sided forecaster.  To  define these distributions, we associate triplets in  $[0,1]^3$ with distribution over  $\zo^2$ in the following way.

Recall that in \cref{subsec:forecasters}, we use \cref{nota:tValDis}, restated below.
\begin{notation}\label{nota:two:ValDis}
	For  $p= (\pa,\pbz,\pbo) \in [0,1]^3$, let $U_p$ denote the random  variable over $\zo^2$ defined by  $\pr{U_p = (x,y)} = \pr{U_{\pa} = x} \cdot \pr{U_{\pbb{x}} = y}$. For $p= (\pa,\pb) \in [0,1]^2$, let $U_p$ denote the random variable $U_{(\pa,\pb,\pb)}$.
\end{notation}

\begin{definition}[Real and forecasted distributions, two-sided case]\label{dfn:two:RealandForcasted}
	For a two-sided forecaster $\Fc$ and an ensemble of finite distributions  $Z = \set{Z_\kappa=(X_\kappa,Y_\kappa,T_\kappa)}_{\pk\in\N}$ over $\zo \times \zo \times\zs$, we define the  real and forecasted distributions $\REAL^{Z,\Fc}$ and  $\FSC^{Z,\Fc}$ by
	\begin{description}
		\item[Real:] $\REAL^{Z,\Fc}_\kappa=(X_\kappa,Y_\kappa,T_\kappa,R_\kappa)$.
		\item[Forecasted:]  $\FSC^{Z,\Fc}_\kappa=(U_{\Fc(T_\kappa;R_\kappa)},T_\kappa,R_\kappa)$.
	\end{description}
	Where $R_\kappa$ is a uniform and independent string  whose length is the (maximal) number of coins used by $\Fc_\kappa$,\footnote{As in the one-sided case, since we only care about \pptm's,  we will implicitly assume that the number of coins used by them on a given security parameter is efficiently computable.} and $\Fc_\kappa(t;r)$ denotes the output of $\Fc_\kappa$ on input $t$ and randomness $r$.
\end{definition}
Namely, $\REAL^{Z,\Fc}$ is just $Z$ concatenated with the randomness of the length used by $\Fc$, where $\FSC^{Z,\Fc}$ is the distribution forecasted by $\Fc$ (given $T$ and $R$ as input).

\paragraph{Indistinguishability.}
Similarly to the one-sided case, the computational distance between the real and  forecasted distribution, measures how well the forecaster realizes the real distribution, from the point of view of a computationally bounded distinguisher.

\begin{definition}[Forecaster indistinguishability,  two-sided case]\label{def:two:Indist}
	A two-sided forecaster $\Fc$ is {\sf $(\rho,\I)$-indistinguishable},  for $\rho>0$ and infinite subset $\I \subseteq \N$, \wrt an ensemble of finite distributions $Z=\set{Z_\kappa=(X_\kappa,Y_\kappa, T_\kappa)}_{\pk\in\N}$ over $\zo \times\zo \times \zs$, if
	$$\REAL^{Z,\Fc} \cindist_{\rho,\I} \FSC^{Z,\Fc}.$$   	
\end{definition}
That is, for sufficiently large $\kappa\in \II$, the forecasted  and real distributions are  $\rho$ indistinguishable for poly-time distinguishers. In \cref{sec:two:ExistenceOfDIst},  we prove our main  result for two-sided forecasters.

\def\thmTwoIndist
{
		For every  ensemble of finite distributions  $Z = \set{Z_\kappa=(X_\kappa,Y_\kappa,T_\kappa)}_{\kappa \in \N}$ over $\zo \times \zo \times \zs$,  $\rho>0$ and an infinite $\I \subseteq\N$, there exists a two-sided forecaster $\Fc$ and an infinite set $\I' \subseteq \I$, such that $\Fc$ is  $(\rho,\I')$-indistinguishable  \wrt $Z$.
}
\begin{theorem}[Existence of indistinguishable two-sided forecaster]\label{thm:two:Indist}
\thmTwoIndist
\end{theorem}

\paragraph{The price of  two-sided forecasters.}
Similarly to the one-sided case, we associate a price function  with  a given  ensemble of finite distributions of the above form and a two-sided forecaster,   which intuitively measures the quality of  the forecaster (a smaller price corresponds to a better forecast).

\begin{notation}\label{not:two:TwotoOne}
	Given a two-sided forecaster $\Fc$ and $i\in \set{1,2,3}$, we let $\Fc^i(t)= \Fc(t)_i$. Given an ensemble of finite distributions $Z = \set{Z_\kappa=(X_\kappa,Y_\kappa,T_\kappa)}_{\pk\in\N}$  over $\zo \times \zo \times \zs$,  let $Z^1 = \set{Z^1_\kappa=(X_\kappa,T_\kappa)}_{\pk\in\N}$, $Z^2 = \set{Z^2_\kappa=((Y_\kappa,T_\kappa) \mid X_\kappa=0)}_{\pk\in\N}$ and $Z^3 = \set{Z^3_\kappa=((Y_\kappa,T_\kappa) \mid X_\kappa=1)}_{\pk\in\N}$.\footnote{Following the convention we coin in \cref{sec:notations}, $Z^2_\kappa$ [\resp  $Z^3_\kappa$] is arbitrarily defined if   $\pr{X_\kappa=0} =0$ [\resp $\pr{X_\kappa=0} =1$].}
\end{notation}
Namely, $\Fc^i$ is the one-sided forecaster induced by $\Fc$ for $Z^i$. The price of a two-sided forecaster \wrt an ensemble of finite distributions $Z$, is defined as the weighted sum of the price of its induced one-sided forecasters \wrt the relevant distributions.

\begin{definition}[Price of a two-sided forecasters]\label{dfn:two-sidedPrice}
 The \remph{price of a two-sided forecaster} \wrt an ensemble of finite distributions $Z=\set{Z_\kappa=(X_\kappa,Y_\kappa, T_\kappa)}_{\pk\in\N}$ over $\zo \times\zo \times \zs$, is defined for $\kappa \in \N$ by
\[\Price_{Z_\kappa}(\Fc_\kappa)= \Price_{Z^1_\kappa}(\Fc^1_\kappa) + \pr{X_\pk=0}\cdot\Price_{Z^2_\kappa}(\Fc^2_\kappa) + \pr{X_\pk=1}\cdot\Price_{Z^3_\kappa}(\Fc^3_\kappa) \]

for $\Price$ being the (one-sided) price function from \cref{dfn:one-sidedPrice}.
\end{definition}

The following relation between price and indistinguishability, proven in \cref{sec:two:DistinguisherGivesImprovedForecaster}, is a main tool in the proof  of \cref{thm:two:Indist}.

\def\DistToPriceLemmaTwo
{
	Let $\Fc$ be a two-sided  forecaster and let  $Z = \set{Z_\kappa=(X_\kappa,Y_\kappa,T_\kappa)}_{\pk\in\N}$ be an  ensemble  of finite distributions over $\zo \times \zo \times\zs$. If there exists a \pptm  $\Dc$ and infinite $\II \subseteq \N$ such that
	\[ \size{\pr{\Dc_\kappa(\REAL^{Z,\Fc}_\kappa)=1} - \pr{\Dc_\kappa(\FSC^{Z,\Fc}_\kappa)=1}} > \rho \]
	for every $\pk\in \II$,	then there exists an infinite subset $\II'\subseteq \II$ and a two-sided forecaster $\Fc'$ such that $\Price_{Z_\pk}(\Fc'_\pk)<\Price_{Z_\pk}(\Fc_\pk)-(\rho/3)^3$ for every $\pk\in \II'$.
}

\begin{lemma}[Distinguishability to   price improvement, two-sided case]\label{cor:one:DistinguisherGivesImprovedForecaster}
	\DistToPriceLemmaTwo
\end{lemma}

\paragraph{Optimal forecasters.}
Roughly speaking, an optimal forecaster \wrt   distribution $Z$, has the lowest  price among all other forecasters \wrt this distribution.  The existence of such forecasters for any ensemble of finite distributions, is the corner stone for the proof of our main result.

\begin{definition}[Optimal forecasters]\label{def:two:optimalForcasters}
A two-sided  forecaster $\Fc$ is  {\sf $(\mu,\I)$-optimal \wrt an ensemble of finite distributions $Z=\set{Z_\kappa=(X_\kappa,Y_\kappa, T_\kappa)}_{\pk\in\N}$ over $\zo \times\zo \times \zs$}, for $\mu > 0$ and infinite $\I \subseteq \N$, if  for every two-sided forecaster $\Fc'$ and   every sufficiently large $\kappa \in \II$, $\Price_{Z_\kappa}(\Fc_\kappa) \le \Price_{Z_\kappa}(\Fc'_\kappa)+\mu$.
	
\end{definition}

The following fact, proven in \cref{sec:two:optimalForcaster}, is a main tool in the proof  of \cref{thm:two:Indist}.

\def\ThmExistOPtimalFOrecaster
{
	For every ensemble of finite distributions  $Z = \set{Z_\kappa=(X_\kappa,Y_\kappa,T_\kappa)}_{\pk\in\N}$ over $\zo \times \zo \times \zs$, $\mu>0$ and infinite $\I\subseteq\N$, there exists a two-sided forecaster $\Fc$ and an infinite set $\I' \subseteq \I$, such that $\Fc$ is  $(\mu,\I')$-optimal \wrt $Z$.
}

\begin{lemma}[Existence of optimal forecaster]\label{lemLtwo:OptimalForcsterExist}
\ThmExistOPtimalFOrecaster
\end{lemma}

\begin{remark}
We emphasize that the proof of \cref{lemLtwo:OptimalForcsterExist} is what restricts us to constant distinguishability error in the main theorem. The rest of the proof goes through for any non-negligible error.
\end{remark}

\subsubsection{Existence of Indistinguishable Forecaster}\label{sec:two:ExistenceOfDIst}
In this section we prove our main result for two-sided forecasters.

\begin{theorem}[Existence of indistinguishable two-sided forecaster, restatement of \cref{thm:two:Indist}]\label{thm:two:IndistRes}
	\thmTwoIndist
\end{theorem}
\begin{proof}
	The proof  follows by the existence of an optimal forecaster, and by the fact that a distinguisher can be used to improve a forecaster.
	
	Let $\mu=(\rho/3)^3$. By \cref{lemLtwo:OptimalForcsterExist} there exists an infinite subset $\II'\subseteq\II$ and a forecaster $\Fc$ that is $(\mu,\I')$-optimal \wrt $Z$. We now claim that $\Fc$ is also $(\rho,\I')$-indistinguishable  \wrt $Z$, as desired.
	
		 Assume toward contradiction, that there exists an infinite subset $\II''\subseteq\II'$ and a \pptm $\Dc$ such that $\size{\pr{\Dc_\kappa(\REAL^{Z,\Fc}_\kappa)=1} - \pr{\Dc_\kappa(\FSC^{Z,\Fc}_\kappa)=1}} > \rho$, for every $\pk\in\II''$. By \cref{cor:one:DistinguisherGivesImprovedForecaster} there exists an infinite subset $\hat{\II}\subseteq\II''$ and a forecaster $\Fh$ such that, $\Price_{Z_\pk}(\Fc_\pk)-\Price_{Z_\pk}(\Fh_\pk)>(\rho/3)^3=\mu$ for every $\pk\in\hat{\I}$. Since $\hat{\I}\subseteq\I'$, this is contradiction to the fact that $\Fc$ is $(\mu,\I')$-optimal.
\end{proof}

\subsubsection{Distinguishability to Price Improvement}\label{sec:two:DistinguisherGivesImprovedForecaster}
In this section we prove the following lemma.
\begin{lemma}[Distinguishability to   price improvement, two-sided case, restatement of \cref{cor:one:DistinguisherGivesImprovedForecaster}]\label{cor:two:DistinguisherGivesImprovedForecasterRes}
	\DistToPriceLemmaTwo
\end{lemma}

We use the following lemma, that  allows us to reduce the proof of the above lemma  to the single-sided case.

\begin{lemma}[Two-sided distinguisher implies one-sided distinguisher]\label{lemma:two:DistinguisherLOneDistinguisher}
	Let $\Fc$ be a two-sided  forecaster, and let  $Z = \set{Z_\kappa=(X_\kappa,Y_\kappa,T_\kappa)}_{\kappa\in\N}$ be an ensemble of finite distributions over $\zo \times \zo \times\zs$. Let $\Fc^1,\Fc^2,\Fc^3$  and  $Z^1,Z^2,Z^3$, be the one-sided forecasters and the ensembles of finite distributions defined according to  \cref{not:two:TwotoOne} \wrt $\Fc$ and $Z$. Assume there exists  \pptm  $\Dc$ and an infinite $\II \subseteq \N$,  such that for every $\kappa \in \II$,
	\[ \size{\pr{\Dc_\kappa(\REAL^{Z,\Fc}_\kappa)=1} - \pr{\Dc_\kappa(\FSC^{Z,\Fc}_\kappa)=1}} > \rho \]
	 Then there exists a \pptm  $\Dc'$ and an infinite subset $\II'\subseteq \II$  such that one of the following hold:
	 \begin{itemize}
	 	\item For every $\pk\in\II'$, $\size{\pr{\Dc'_\pk(\REAL^{Z^{1},\Fc^{1}}_\kappa)=1 } - \pr{\Dc'_\kappa(\FSC^{Z^{1},\Fc^{1}}_\kappa)=1}}  > \rho/3$.
	 	
	 	\item There exists $b\in\zo$, such that for every $\pk\in\II'$,
	 	$$\size{\pr{\Dc'_\pk(\REAL^{Z^{2+b},\Fc^{2+b}}_\kappa)=1 } - \pr{\Dc'_\kappa(\FSC^{Z^{2+b},\Fc^{2+b}}_\kappa)=1}}\cdot\pr{X=b}  > \rho/3.$$
	 \end{itemize}
 where the distributions $\REAL^{Z^i,\Fc^i}$ and $\FSC^{Z^i,\Fc^i}$ above, are the ``one-sided'' distributions according to \cref{dfn:one:RealandForcasted}.
\end{lemma}

\cref{lemma:two:DistinguisherLOneDistinguisher} is proven below, but we first use it for proving \cref{cor:two:DistinguisherGivesImprovedForecasterRes}.

\newcommand{\Ih}{\widehat{\I}}
\begin{proof}[Proof of \cref{cor:two:DistinguisherGivesImprovedForecasterRes} ]
	
		By \cref{lemma:two:DistinguisherLOneDistinguisher}, there exists a \pptm $\Dc'$, an infinite set $\I'\subseteq\II$ and a fixed $\is\in[3]$, such that  for every $\pk\in\I'$:
			\begin{align}\label{eq:two:DistinguisherGivesImprovedForecasterRes}
			\size{\pr{\Dc'_\pk(\REAL^{Z^\is,\Fc^\is}_\kappa)=1 } - \pr{\Dc'_\kappa(\FSC^{Z^\is,\Fc^\is}_\kappa)=1}}  > \rho/3
			\end{align}
		   and  if $\is\in\set{2,3}$, then also
		 \begin{align}\label{eq:two:DistinguisherGivesImprovedForecasterRes2}
		 \pr{X_\pk=(\is-2)}>\rho/3
		 \end{align}
		 By \cref{lemma:DistinguisherImplyImprove} and \cref{eq:two:DistinguisherGivesImprovedForecasterRes}, there exist a one-sided forecaster $\Fh$ and an infinite set $\Ih\subseteq\I'$, such that for every $\pk\in\Ih$:
		
		 \begin{align*}
		 \Price_{Z^\is_\pk}(\Fc^\is_\pk)-\Price_{Z^\is_\pk}(\Fh_\pk)>(\rho/3)^2
		 \end{align*}
		
Consider the  two-sided forecaster $\Fc'$ resulting by replacing $\Fc^\is$ with $\Fh$. That is,   $\Fc'(t) = (\Fc'^1(t),\Fc'^2(t),\Fc'^3(t))$, for $\Fc'^i  = \Fh$ if $i=\is$, and $\Fc'^i  = \Fc^i$  otherwise. The definition of the price function yields the following for every $\pk\in\Ih$:

If $\is=1$, then
			\begin{align*}
			\Price_{Z_\pk}(\Fc'_\pk)-\Price_{Z_\pk}(\Fc_\pk)=      \Price_{Z^1_\pk}(\Fc^1_\pk)-\Price_{Z^1_\pk}(\Fh_\pk) >   ({\rho}/{3})^2
				\end{align*}
and if  $\is\in\set{2,3}$,  then
		\begin{align*}
		\Price_{Z_\pk}(\Fc'_\pk)-\Price_{Z_\pk}(\Fc_\pk)&=      \pr{X_\pk=(\is-2)}\cdot(\Price_{Z^\is_\pk}(\Fc^\is_\pk)-\Price_{Z^\is_\pk}(\Fh_\pk))\\
		& >   \pr{X_\pk=(\is-2)}\cdot({\rho}/{3})^2 \\
		&\ge({\rho}/{3})^3,
		\end{align*}
where the last inequality holds by \cref{eq:two:DistinguisherGivesImprovedForecasterRes2}. This concludes the proof.
\end{proof}

\paragraph{Proof of \cref{lemma:two:DistinguisherLOneDistinguisher}.}
\begin{proof}[Proof of \cref{lemma:two:DistinguisherLOneDistinguisher}]

We use the following    algorithm to  define three different distinguishers, and then prove that at least one of them can serve as $\Dc'$.
	
	\begin{algorithm}[$\Ac$]\label{alg:sim}
		\item Input: Security parameter $1^\pk$ and $(v,t,r) \in \zo \times \zs \times \zs$.
		
		\item Operation: If $v=0$, output $\Fc^2_\pk(t,r)$, else,  output $\Fc^3_\pk(t,r)$.
	
	\end{algorithm}

	By definition,
	\begin{align*}
	\FSC^{Z,\Fc}_\pk=(X'_\pk,U_{\Ac_\kappa(X'_\pk,T_\pk,R_\pk)},T_\pk,R_\pk)
	\end{align*}
	for $X_\kappa' = U_{\Fc^1(T_\pk;R_\pk)}$. Let $\Dc^1_\pk(v,t,r)=\Dc_\pk(v,U_{\Ac_\pk(v,t,r)},t,r)$ , $\Dc^2_\pk(v,t,r)=\Dc_\pk(0,v,t,r)$ and $\Dc^3_\pk(v,t)=\Dc_\pk(1,v,t,r)$. We conclude the proof using  the following claim, proven below.

	\begin{claim}\label{claim:2-sided Distinguisher imply 1-sided Distinguisher}
		Let $\kappa \in \N$ be such that
		$\size{\pr{\Dc_\kappa(\REAL^{Z,\Fc}_\kappa)=1} - \pr{\Dc_\kappa(\FSC^{Z,\Fc}_\kappa)=1}}  > \rho$. Then (at least)  one of the following holds,
		\begin{enumerate}
			 	\item $\size{\pr{\Dc^1_\pk(\REAL^{Z^{1},\Fc^{1}}_\pk)=1} - \pr{\Dc^1_\kappa(\FSC^{Z^{1},\Fc^{1}}_\pk)=1}}  > \rho/3$,
		
			 	\item $\size{\pr{\Dc^{2}_\pk(\REAL^{Z^{2},\Fc^{2}}_\kappa)=1 } - \pr{\Dc^{2}_\kappa(\FSC^{Z^{2},\Fc^{2}}_\kappa)=1}}\cdot\pr{X=0}  > \rho/3$, or
			 	
			 	\item $\size{\pr{\Dc^{3}_\pk(\REAL^{Z^{3},\Fc^{3}}_\kappa)=1 } - \pr{\Dc^{3}_\kappa(\FSC^{Z^{3},\Fc^{3}}_\kappa)=1}}\cdot\pr{X=1}  > \rho/3$.
		\end{enumerate}
	\end{claim}

	 By  \cref{claim:2-sided Distinguisher imply 1-sided Distinguisher} and the Pigeonhole principle, there exists $i\in[3]$ and an infinite set $\II'\subseteq \II$, such that  $\Dc^i$ satisfies the \ith  item  in the claim for every $\pk\in\II'$. Thus, the proof follows by taking $\Dc'=\Dc^i$.
\end{proof}

\paragraph{Proof of  \cref{claim:2-sided Distinguisher imply 1-sided Distinguisher}.}
\begin{proof}[Proof of \cref{claim:2-sided Distinguisher imply 1-sided Distinguisher}]
	
	Fix $\pk\in \N$ that satisfies the condition of the claim, and omit it from the following text to avoid clutter. By definition,
	\begin{align}\label{eq:distance for 2-sided Distinguisher}
	\size{ \pr{\Dc((X,Y,T,R)= \REAL^{Z,\Fc})=1}-\pr{\Dc((X',U_{\Ac(X',T,R)},T,R) = \FSC^{Z,\Fc})=1}} > \rho
	\end{align}

	Consider the hybrid distribution
	$$H=(X,U_{\Ac(X,T,R)},T,R)$$ resulting from replacing  the ``forecasted'' $X'$ in $\FSC^{Z,\Fc}$ with the ``real'' value  $X$.
	 By \cref{eq:distance for 2-sided Distinguisher},
	\begin{align*}
	\size{\pr{\Dc(\REAL^{Z,\Fc})=1}-\pr{\Dc(H)=1} +\pr{\Dc(H)=1}-\pr{\Dc(\FSC^{Z,\Fc})=1}}> \rho
\end{align*}
	and thus either
	\begin{align}
	&\size{\pr{\Dc(H)=1}-\pr{\Dc(\FSC^{Z,\Fc})=1}}>\rho/3, \text{ or}\label{eq:dist for x} \\
	&\size{\pr{\Dc(\REAL^{Z,\Fc})=1}-\pr{\Dc(H)=1}}>2\rho/3 \label{eq:dist for y}
	\end{align}

	Suppose  \cref{eq:dist for x} holds. By definition, $ \Dc^1(X,T,R) \equiv \Dc(H)$ and $ \Dc^1(X',T,R) \equiv \Dc(\FSC^{Z,\Fc})$. Thus,
	$\size{ \pr{\Dc^1(X,T,R)=1}-\pr{\Dc^1(X',T,R)=1}}>\rho/3$, which concludes the proof since $(X',T,R)=\FSC^{Z^{1},\Fc^{1}}$ and $(X,T,R) =\REAL^{Z^{1},\Fc^{1}}$.
	
	Suppose now that \cref{eq:dist for y} holds. It follows that for some  $b\in \zo$
	\begin{align*}
	\pr{X=b}\cdot\size{\pr{\Dc(b,Y,T,R)=1\mid X=b}-\pr{\Dc(b,U_{\Ac(b,T,R)},T,R)=1\mid X=b}} >\rho/3
	\end{align*}
	
	Since, by definition,
	\begin{align*}
	(U_{\Ac(b,T,R)},T,R)|_{X=b} \equiv (U_{\Fc^{2+b}(T;R)},T,R)|_{X=b} \equiv \FSC^{Z^{2+b},\Fc^{2+b}}
	\end{align*}
    and
    	\begin{align*}
    	(Y,T,R)|_{X=b} \equiv \REAL^{Z^{2+b},\Fc^{2+b}}
    \end{align*}	
    	
	It follows that $\pr{X=b}\cdot\size{\pr{\Dc^{2+b}(\REAL^{Z^{2+b},\Fc^{2+b}})}-\pr{\Dc^{2+b}(\FSC^{Z^{2+b},\Fc^{2+b}})}}>\rho/3$, concluding the proof.
\end{proof}		

\subsubsection{Existence of Optimal Forecasters}\label{sec:two:optimalForcaster}
In this section we prove the following lemma.
\begin{lemma}[Existence of optimal forecaster, restatement of \cref{lemLtwo:OptimalForcsterExist}]\label{lemLtwo:OptimalForcsterExistRes}
	\ThmExistOPtimalFOrecaster
\end{lemma}
\begin{proof}
Let $\cF$ denote the set of all  forecasters. Consider the following iterative process:
\begin{description}
\item[Initialization:] We start by picking some $\Fc^{(1)} \in \cF$, and let $\II_1=\I$, and $\nu_1=2$.

\item[Step $i$:] (start with Step $1$)
\begin{enumerate}
\item At the beginning of step $i$ we hold $\Fc^{(i)} \in \Fc$ and an infinite set $\II_i \subseteq \N$, such that  $\Price_{Z_\kappa}(\Fc^{(i)}_\kappa) \le \nu_i$  for every $\kappa \in \II_i$. (Note that this holds trivially for $i=1$, because the price function of a forecaster is bounded from above by $2$).

\item If  exists $\Fh \in \cF$ and an infinite subset $\II' \subseteq \II_i$, such that
    \[ \Price_{Z_\kappa}(\Fh_\kappa) < \Price_{Z_\kappa}(\Fc^{(i)}_\kappa) -\mu, \]
    for every $\kappa \in \II'$, set $\Fc^{(i+1)}=\Fh$, $\nu_{i+1}=\nu_i - \mu$ and $\II_{i+1}=\II'$, and continue to step $i+1$.

     Note that we indeed have that for every $\kappa \in \II_{i+1}$,
    \[\Price_{Z_\kappa}(\Fc^{(i+1)}_\kappa) < \Price_{Z_\kappa}(\Fc^{(i)}_\kappa) - \mu \le \nu_i - \mu =\nu_{i+1}. \]
    Therefore we meet the requirement at the beginning of step $i+1$.
\item Otherwise, we have that for every $\Fh \in \cF$, there are only finitely many $\kappa \in \II_i$, for which
    \[ \Price_{Z_\kappa}(\Fh_\kappa) < \Price_{Z_\kappa}(\Fc^{(i)}_\kappa) -\mu. \]
    This means that for every sufficiently large $\kappa \in \II_i$,
    \[ \Price_{Z_\kappa}(\Fh_\kappa) \ge  \Price_{Z_\kappa}(\Fc^{(i)}_\kappa) -\mu. \]

    It follows that $\Fc^{(i)}$ is $(\mu,\II_i)$-optimal \wrt $Z$, and we obtain an optimal forecaster.
\end{enumerate}
\end{description}

Noting that at every step $i$, if we continue to the next step, then $\nu_{i+1} \le \nu_{i} - \mu$. However, at every step $i$, it is trivial that $\nu_i \le 2$. This is because, the price of of a forecaster is bounded by $2$. It follows that after at most $2/\mu$ iterations, we will obtain an infinite  set $\II'$ and a  forecaster $\Fc$ that is $(\mu,\II')$-optimal \wrt $Z$, as required.
\end{proof}

\begin{remark}[on the generality of the above argument]
It is instructive to note that we have used no specific properties of the price function or of the set $\cF$ and the argument will work just the same for every choice of price function, and every class $\cF$ of functions.
\end{remark}

\section{Correlated Forecaster to Key Agreement}\label{sec:ForecastingToKA}
In this section we show how to use a protocol that has a correlated indistinguishable forecaster  to construct a key-agreement protocol. The core of the reduction is a new information theoretic key-agreement protocol, that we can apply in the computational setting using an indistinguishable forecaster (recall that this approach is explained in the introduction).

\subsection{Non-oblivious Key Agreement  from Correlated Distributions}\label{sec:nonObliviousKA}

Key-agreement protocols in the information theoretic setting assume that  two (honest) parties $\Ac$ and $\Bc$, and an adversary (eavesdropper) $\Ec$, receive (possibly correlated) random variables $X$, $Y$ and $T$, respectively. The goal of the parties  is to interact, so that their final outputs will be identical, and statistically close to a uniform distribution even conditioned on $T$ and the transcript of their interaction. Note that in this setting, the honest parties  do not see $T$. This is in contrast to the computational setting, where we imagine that $T$ is the transcript of some earlier protocol, and is available to the honest parties.

We will now consider an information theoretic setting which we refer to as a non-oblivious. In this setting, the honest parties $\Ac$ and $\Bc$ receive inputs $X'$ and $Y'$ respectively, and in addition they also receive $T$ (for this reason we refer to this setup as a non-oblivious, since the parties are not oblivious to the transcript $T$). The adversary $\Ec$ is unbounded, and receives (only) $T$. Loosely speaking, this setting corresponds to the following setup:  a protocol $\pi$ was run on input $1^\kappa$ generating transcript $T$, and the parties'  outputs are $X$ and $Y$ respectively. We consider a simulation of that protocol (in the sense of Section \ref{sec:Classification}) that produces a triplet $(X',Y',T)$ that is somewhat indistinguishable from $(X,Y,T)$. Indeed, in this information theoretic setting, $\Ac$ and $\Bc$ receive $X'$ and $Y'$ respectively, and also receive access to $T$. The adversary $\Ec$ receives $T$.
There are several advantages in considering this scenario:
\begin{itemize}
\item $(X',Y')$ has information theoretic uncertainty given $T$, and so we can work in an information theoretic setting where $\Ec$ is unbounded.
\item The honest parties see $T$.
\item Moreover, the honest parties have access to a (\ppt) forecaster, which given $t$, allows them to compute all probabilities in the probability space $(X',Y')|_{T=t}$.
\end{itemize}

We now describe a key-agreement protocol in this setting. More precisely, in the protocol below, in addition to their inputs, parties are given access to a function $f:\zs \to [0,1]^3$ which on input $t$, produces a description of the probability space $(X',Y')|_{{T=t}}$. We will show that this protocol is a key-agreement that has perfect secrecy, and agreement that depends on the ``correlation distance'' of the forecasted distribution. A precise statement appears below. Later, we will ``pull back'' this protocol to the computational world, using an indistinguishable forecaster.

\begin{protocol}[Non-oblivious key-agreement protocol $\KA^f = (\Ac,\Bc)$]\label{proto:info-WKA}
	\item Common input $t\in \zs$.
	\item \Ac's private input: $x\in\zo$.
	\item \Bc's private input: $y\in\zo$.
	
	\item Oracle:  function $f\colon \zo^\ast\mapsto[0,1]^3$.
	
	\item Operation:
	
	\begin{enumerate}	
		\item Both parties compute $p=(p_1,p_2,p_3) = f(t)$.

		\item \Ac  samples $x'\la  U_{p_1}$.\label{proto:info-WKA:step1}
		
		\item If $x\neq x'$,  $\Ac$ outputs $x$. Otherwise it outputs a random bit.
		\item  $\Bc$ outputs $y$  if $p_3 >p_2$, and $(1-y)$ otherwise.
	\end{enumerate}
\end{protocol}
The following lemma relates the quality of the above protocol, as key agreement, to the ``correlation'' of its inputs  distribution.

Recall, that $U_{p =(p_1,p_2,p_3)}$ is a random variable over $\zo^2$ distributed according to $p$ (\ie $\pr{U_p = (x,y)} = \pr{U_{p_1} = x} \cdot \pr{U_{p_{x+2}} = y}$),   and that, see \cref{nota:ProdOfDis},  $\Prod(p)$ is  the description of the product of $U_p$ marginals (\ie $\Prod(p) = (p_1, (1-p_1) \cdot p_2+ p_1 \cdot p_3)$).
\begin{lemma}\label{lemma:info-ka}
	Let  $Z=(X,Y,T)$ be a triplet distributed over $\zo\times\zo\times\zs$, and let $f\colon \zs\mapsto \zo^3$ be such that
	$(X,Y)|_{ f(T) =t} \equiv U_{f(t)}$ for every $t\in \Supp(T)$.
Let  $(\Ac,\Bc)=\KA^f$ be the protocol as specified in \cref{proto:info-WKA} and let $\eta = \SD((X,Y,T),(U_{\Prod(f(T))},T))$. Then
	\begin{description}
		\item[Agreement:] $\pr{(\Ac(X),\Bc(Y))(T) = (b,b) \mbox{ for some $b\in \zo$}}=\half+\eta/2$.
		\item[Secrecy:] $\pr{(\Ac(X),\Bc(Y))(T) = (1,\cdot)\mid T=t}=  1/2$, for every $t\in\Supp(T)$.
	\end{description}
\end{lemma}
\begin{proof}
Let $X'$ denote the value of $x'$  sampled by $\Ac(X,T)$ (\stepref{proto:info-WKA:step1}), and let $(O_\Ac,O_\Bc)= (\Ac(X),\Bc(Y))(T)$. We use the following claims, proven below.

\begin{claim}\label{claim:info-ka:1}
	$\pr{O_\Ac=1\mid T=t}=  1/2$ for every $t\in\Supp(T)$.
\end{claim}	


\begin{claim}\label{claim:info-ka:33}
	$\pr{O_\Ac = O_\Bc \land  X= X'} =\pr{X= X'}/2$.
\end{claim}

\begin{claim}\label{claim:info-ka:4}
 $\pr{O_\Ac = O_\Bc \land  X\ne X'} =\half(\pr{X\ne X'} +  \eta)$.
 \end{claim}
The secrecy part immediately follows from \cref{claim:info-ka:1}. For the agreement part, using \cref{claim:info-ka:33,claim:info-ka:4} we get that
	\begin{align*}
	\pr{O_\Ac = O_\Bc}=&\pr{O_\Ac = O_\Bc \land  X= X'}+ \pr{O_\Ac = O_\Bc \land  X\ne X'}\\
					  =& \half(1-\pr{X\ne X'} + \pr{X\ne X'} +  \eta)\nonumber\\
					  =& \half+\eta/2.
	\end{align*}
\end{proof}
We now proceed to  proving  \cref{claim:info-ka:1,claim:info-ka:33,claim:info-ka:4}.

\begin{proof}[Proof of \cref{claim:info-ka:1}]
	Fix $t\in\Supp(T)$. Since  $\Ac$ outputs a uniform bit if $X= X'$, it holds that
	\begin{align*}
	\pr{O_\Ac=1\mid T=t,X= X' } = \half
	\end{align*}
  Hence, we can assume \wlg that $\pr{X\ne X' \mid T=t}\ne 0$, as otherwise by the above equality we are done. Note that
	\begin{align*}
	\pr{O_\Ac=1\mid T=t,X\ne X' } &= 	\pr{X=1\mid T=t,X\ne X' }\\
	&= \frac{\pr{X=1 \land X'=0 \mid T=t}}{\pr{X\neq X' \mid T=t }}\nonumber\\
		&= \frac{\pr{X=1 \land X'=0 \mid T=t}}{2\pr{X=1 \land X'=0 \mid T=t}}\nonumber\\
	&= \half,\nonumber
	\end{align*}
where  the penultimate equality holds since $\pr{X=1 \land X'=0 \mid T=t} = f(t)_1 \cdot (1- f(t)_1) =  \pr{X=0 \land X'=1 \mid T=t}$. 	It follows that,
	\begin{align*}
	\lefteqn{\pr{O_\Ac=1\mid T=t}}\\
	&= \pr{X\ne X' \mid T=t,} \cdot \pr{O_\Ac=1\mid T=t,X\ne X' } + \pr{X= X'\mid T=t,} \cdot \pr{O_\Ac=1\mid T=t,X= X' }\\
	& = \pr{X\ne X'\mid T=t,} \cdot\half+ \pr{X= X'\mid T=t,} \cdot \half = \half.
	\end{align*}
\end{proof}

\begin{proof}[Proof of \cref{claim:info-ka:33}]
Holds since  $\Ac$ outputs a uniform bit if $X= X'$.
\end{proof}

\begin{proof}[Proof of \cref{claim:info-ka:4}]
	We will show that for every $t\in \Supp(T)$,
	\begin{align*}
	\pr{O_\Ac = O_\Bc \land  X\ne X'\mid T=t} =  f(t)_1\cdot (1-f(t)_1)\cdot(1 +\size{f(t)_2-f(t)_3})
	\end{align*}
	We assume \wlg that $\pr{X\ne X' \mid T=t}\ne 0$, as otherwise $f(t)_1\in\zo$ and the above equality trivially holds. Fix $t\in\Supp(T)$, and let $p = (p_1,p_2,p_3)=f(t)$, we want to calculate $\pr{O_\Ac=O_\Bc\mid X\ne X',T=t}$.  The proof continues according to whether $p_3>p_2$.

	 Assuming $p_3>p_2$, then $\pr{O_\Ac=O_\Bc\mid X\ne X',T=t}=\pr{X=Y\mid T=t, X\ne X'}$. Thus
	\begin{align*}
	\pr{X=Y\mid T=t, X\ne X'} =&\pr{X=1\mid T=t, X\ne X'}\cdot\pr{Y=1\mid T=t,X=1} \\&+ \pr{X=0\mid T=t, X\ne X'}\cdot\pr{Y=0\mid T=t,X=0}\nonumber\\
	=&\half\cdot\pr{Y=1\mid T=t,X=1} + \half\cdot\pr{Y=0\mid T=t,X=0}\nonumber\\
	=&\half\cdot p_3 +\half\cdot(1-p_2)\nonumber\\
	=&\half(1+(p_3-p_2)). \nonumber
	\end{align*}
	
		Assuming $p_3 \le p_2$, then $\pr{O_\Ac=O_\Bc\mid X\ne X',T=t}=\pr{X\ne Y\mid T=t, X\ne X'}$. Thus
	\begin{align*}
	\pr{X\ne Y\mid T=t, X\ne X'} =&\pr{X=1\mid T=t, X\ne X'}\cdot\pr{Y=0\mid T=t,X=1} \\
	&+\pr{X=0\mid T=t, X\ne X'}\cdot\pr{Y=1\mid T=t,X=0}\nonumber\\
	=&\half\cdot\pr{Y=0\mid T=t,X=1} + \half\cdot\pr{Y=1\mid T=t,X=0}\nonumber\\
	=&\half\cdot (1-p_3) +\half\cdot p_2\nonumber\\
	=&\half(1+(p_2-p_3)).\nonumber
	\end{align*}
	Putting it together, $\pr{O_\Ac=O_\Bc\mid T=t, X\ne X'}  =\half(1+\size{p_2-p_3})$. Since, $\pr{X\ne X'\mid T=t}=2\cdot p_1\cdot(1-p_1)$, it follows that
	\begin{align}\label{eq:info-ka:4}
	\pr{O_\Ac = O_\Bc \land  X\ne X'\mid T=t}  = p_1\cdot(1-p_1)\cdot(1+\size{p_2-p_3})
	\end{align}
	We conclude that
	\begin{align*}
	\pr{O_\Ac = O_\Bc \land  X\ne X'}&=\eex{t\la T}{O_\Ac = O_\Bc \land  X\ne X\mid T=t}\\
					&=\eex{t\la T}{ f(t)_1\cdot (1-f(t)_1)\cdot(1+\size{f(t)_2-f(t)_3})}\\
					&=\eex{t\la T}{ f(t)_1\cdot (1-f(t)_1) +  f(t)_1\cdot (1-f(t)_1)\cdot\size{f(t)_2-f(t)_3}}\\
					&=\eex{t\la T}{ f(t)_1\cdot (1-f(t)_1)} + \eex{t\la T}{ f(t)_1\cdot (1-f(t)_1)\cdot\size{f(t)_2-f(t)_3}}\\
					&=\half \cdot \pr{X\ne X'} + \eex{t\la T}{ f(t)_1\cdot (1-f(t)_1)\cdot\size{f(t)_2-f(t)_3}}\\
				    &=\half \cdot \pr{X\ne X'} + \eta/2.
	\end{align*}
	The second equality is by \cref{eq:info-ka:4} and the last  one by   \cref{claim:info-ka:3}, given below.
\end{proof}

\begin{claim}\label{claim:info-ka:3}
	$\eex{t\la T}{(1-f(t)_1)\cdot f(t)_1\cdot\size{f(t)_2-f(t)_3}} = \mu/2$.
\end{claim}

\begin{proof}
	Since $ \eta = \SD((U_{f(T)},T),(U_{\Prod(f(T))},T)) = \eex{t\la T}{\SD(U_{f(t)},U_{\Prod(f(t))})}$,
	it suffices to prove that
	\begin{align}\label{eq:info-ka:3}
	\SD(U_{f(t)},U_{\Prod(f(t))})= 2\cdot f(t)_1\cdot (1-f(t)_1)\cdot \size{f(t)_2-f(t)_3}
	\end{align}
	 for every $t\in\Supp(T)$.

	 Fix such $t$ and  let $p=(p_1,p_2,p_3)= f(t)$, let  $q = p_1 p_3 + (1-p_1) p_2$,  let  $(X_t,Y_t)  = U_{p}$ and   $(X'_t,Y'_t)  = U_{\Prod(p)}$.   We assume \wlg that $p_1 \in (0,1)$, as otherwise  \cref{eq:info-ka:3}  holds trivially.  Compute
	\begin{align*}
	\SD((X_t,Y_t)|_{X_t=0},(X_t',Y_t')|_{X_t'=0})&=\SD((0,U_{p_2},t),(0,U_q,t))  \\
	&=\size{p_2-q} \nonumber\\
	&=\size{p_2-p_1\cdot p_3 -(1-p_1)\cdot p_2}\nonumber\\
	&=p_1\cdot\size{p_2-p_3}. \nonumber
	\end{align*}
	And similarly,
	\begin{align*}
	\SD((X_t,Y_t)|_{X_t=1},(X_t',Y_t')|_{X_t'=1})&=(1-p_1)\cdot\size{p_2-p_3}
	\end{align*} 									
	Since $X_t \equiv X'_t$, by \cref{lemma:SD for joint} we conclude that $\SD((X_t,Y_t),(X_t',Y_t'))=2\cdot p_1\cdot(1-p_1)\cdot\size{p_2-p_3}$.
\end{proof}

\begin{remark}
We remark that \cref{proto:info-WKA} is noninteractive and does not require communications between the parties. Note however that our final key-agreement protocol also includes amplification by Hollenstein \cite{Holenstein06b} which is interactive, and so the final protocol is interactive.
\end{remark}

\subsection{Key Agreement from Correlated Protocols}
In this section  we invoke   \cref{proto:info-WKA} on  a distribution induced by a protocol outputs and transcript, using the  forecaster for this distribution as the oracle $f$ used by \cref{proto:info-WKA}. The  resulting  protocol  inherits  the forecasted distribution indistinguishability and correlation, with small losses that depend on the indistinguishability parameter $\rho$.

Given a protocol $\pi$ and a forecaster \Fc for $\pi$, consider the following protocol



\begin{protocol}[key agreement protocol $\KA^{\Fc,\pi,m} = (\Ac,\Bc)$]\label{proto:compu-WKA}
	
	 \item Parameters: security parameter $1^\secParam$.

	\item Oracles: forecaster $\Fc$, next message  function of a   two-party single output protocol $\pi=(\Ah,\Bh)$ and a function $m\colon \N \mapsto \N$.

	\item Operation:
	\begin{enumerate}
		\item Parties interact in a random executions of $\pi(1^\secParam)$, with  $\Ac$ and $\Bc$ taking the role of $\Ah$ and $\Bh$, respectively. Let $(x,y,t)$, be the local outputs of $\Ac$ and $\Bc$, and the protocol transcript.

		\item \Ac sample a uniform string $r\la \zo^{m(\kappa)}$ and sends it to \Bc.

		\item The parties $\Ac$ and $\Bc$  interact in  $(\At(x),\Bt(y))(1^\pk,t,r)$, for $ (\At,\Bt)= \KA^{\Fc}$ being according to \cref{proto:info-WKA}, and party $\Ac$ plays the role of $\At$, and party $\Bc$ plays the role of $\Bt$.
		
		The parties  output their outputs in the above execution.
	\end{enumerate}
\end{protocol}


\begin{lemma}[Weak key-agreement protocol from correlated protocols]\label{lemma:WKA from correlation}
	Let $\pi$ be a \ppt two-party single-bit output protocol,  let $\Fc$ be a forecaster and let $m\in \poly$ be a bound on number of coins used by $\Fc$  on transcripts of  $\pi(1^\kappa)$ and let $Z_\kappa= (X_\kappa,Y_\kappa,T_\kappa)$ be the distribution of the parties' output and protocol transcripts induce by a random execution of $\pi(1^\kappa)$. Let $\I \subseteq \N$ and  $\rho,\eta>0$ be such that $\Fc$ is $(\rho,\I)$-indistinguishable and $(\eta,\I)$-correlated \wrt $Z= \set{Z_\kappa}_{\kappa\in \N}$,   then protocol $\KA^{\Fc,\pi,m}$, defined in \cref{proto:compu-WKA}, is an $(\rho,\eta/2-\rho)$-key agreement-protocol in $\II$.
\end{lemma}

 \cref{lemma:WKA from correlation} is proven below, but we first use it for proving  the main result of this section.
\begin{theorem}[Key-agreement from correlated protocols, restatement of \cref{thm:fulKAfromCorrelation}]\label{thm:fulKAfromCorrelationRes}
	\CorToKAThm
\end{theorem}
\begin{proof}
	The proof directly follows from \cref{lemma:WKA from correlation,thm:KeyAggAmp}.
\end{proof}

\begin{proof}[ Proof of \cref{lemma:WKA from correlation}]

Let $\KA^{\INFO}= (\Ac^\INFO,\Bc^\INFO)$ be the protocol defined by $\KA^{\INFO}_\pk = \KA^{\Fc_\pk}$, for $\KA^{\Fc}$ being according to   \cref{proto:info-WKA}.  Let  $\Zt = \set{\Zt_\pk = (\Xt_\pk,\Yt_\pk,T_\pk,R_\pk)= \FSC^{\pi,\Fc}_\pk}_{\kappa\in \N}$ and let $Z=\set{Z_\pk = (X_\pk,Y_\pk,T_\pk,R_\pk) = \REAL^{\pi,\Fc}_\pk}_{\kappa\in \N}$ be the real and forecasted distribution of $\Fc$ \wrt $Z$ (see \cref{dfn:RealandForcasted}). For $\pk \in \I$, let $(\Ot^\Ac_\pk,\Ot^\Bc_\pk)$ denote the parties' output in a random execution of  $(\Ac^\INFO(\Xt_\pk),\Bc^\INFO(\Yt_\pk))(1^\pk,T_\pk,R_\pk)$. By  \cref{lemma:info-ka},
 \begin{itemize}
 	\item $\pr{\Ot^\Ac_\pk = \Ot^\Bc_\pk}\ge \half+\eta/2$, and
 	
 	\item $\pr{\Ec(T_\pk,R_\pk)=\Ot^\Ac_\pk}=  1/2$  for every (even unbounded)  algorithm $\Ec$.
 \end{itemize}
Now let $(O^\Ac_\pk,O^\Bc_\pk)$ denote the parties' output in a random execution of  $(\Ac^\INFO(X_\pk),\Bc^\INFO(Y_\pk))(1^\kappa,T_\pk,R_\pk)$. Since $\KA^{\INFO}$ can be computed efficiently (recall that   $\Fc$ is \ppt), and since, by definition,  $Z \cindist_{\rho,\II}\Zt$, it follows that

 \begin{itemize}
 	\item $\pr{O^\Ac_\pk = O^\Ac_\pk}\ge\half+\eta/2-\rho$, for  large enough $\kappa\in \I$, and
 	\item  For every \ppt  $\Ec$ it holds that $\pr{\Ec(T_\pk,R_\pk)=O^\Ac_\pk}\le  1/2 +\rho$, for  large enough $\kappa\in \I$.
 \end{itemize}
     Indeed,  otherwise there exists a \ppt algorithm $\Ec$  that distinguishes between the real and the forecasted distributions with advantage greater than $\rho$, contradicting the fact that $\Fc$ is a $(\rho,\II)$-forecaster \wrt $Z$.

      Let $\KA^{\COMP}=\KA^{\Fc,\pi,m}$  be the (``computational'') protocol defined in  \cref{proto:compu-WKA}. Noting that the transcript and outputs induced by  a random execution of $\KA^{\COMP}(1^\kappa)$ are identical to that of   $(\Ac^\INFO(X_\kappa),\Bc^\INFO(Y_\kappa))(1^\kappa,T_\kappa,R_\kappa)$, yields the proof.
\end{proof}

\newcommand{\DP}{{\sf DP}\xspace}
\newcommand{\EDP}{{\sf EDP}\xspace}
\newcommand{\private}{{\sf private}\xspace}
\newcommand{\hPi}{{\widehat{\pi}}}
\newcommand{\hAc}{{\widehat{\Ac}}}
\newcommand{\hBc}{{\widehat{\Bc}}}
\newcommand{\correct}{{\sf correct}\xspace}
\renewcommand{\io}{{\sf io}\xspace}
\newcommand{\teps}{{\widetilde{\eps}}}
\newcommand{\tdelta}{{\widetilde{\delta}}}

\section{ Non-Trivial Differentially Private XOR  Implies Key Agreement}\label{sec:DPXRtoKA}
In this section we  use our  classification from \cref{sec:Classification}, to prove that a non-trivial differentially private protocol for computing XOR, implies the existence of a  key-agreement protocol.  In \cref{sec:EDP} we extend the reduction for protocols whose privacy guarantee only assumed  to hold against  \emph{external} observers.

\paragraph{Notation.}
	We introduce some new notation to be used for with input protocols.
	Given a two-party protocol  $\pi= (\Ac,\Bc)$, $\Pc\in \set{\Ac,\Bc}$ and $z\in \zs$,  let  $(\trans_\pi(z),\out^\Pc_\pi(z),\view_\pi^\Pc(z))$, denote the transcript, $\Pc$'s output and $\Pc$'s view receptively,  in a random execution of $\pi(z)$.

\paragraph{Differential privacy.}
Since the  focus of this result  is on single bit input  protocol, we only define differential privacy for such protocols. Also, since we are in the computational setting,  we only define the notion for efficient  distinguishers.

\begin{definition}[$(\eps,\delta)$-differential privacy]\label{def:DP}
A single-bit input two-party protocol $\pi = (\Ac,\Bc)$ is {\sf $(\eps,\delta)$-differentially private, denoted $(\eps,\delta)$- \DP},  \wrt  $\eps,\delta \colon \N \mapsto \R^+$, if  for  any \ppt distinguisher $\Dc$ and $x\in\zo$, for all but finitely many $\kappa$'s it holds that

$$ \pr{\Dc(\view^{\Ac}_{\pi}(1^\kappa,x,0))= 1}  \in  e^{\pm \eps(\kappa)} \cdot \pr{\Dc(\view^{\Ac}_{\pi}(1^\kappa,x,1))= 1}  \pm \delta(\kappa)$$
and similarly for any $y\in \zo$:
$$ \pr{\Dc(\view^{\Bc}_{\pi}(1^\kappa,0,y))= 1}  \in  e^{\pm \eps(\kappa)}\cdot  \pr{\Dc(\view^{\Bc}_{\pi}(1^\kappa,1,y))= 1}  \pm\delta(\kappa)$$
\end{definition}
Namely,  an adversary seeing the view of one of the parties, cannot tell the other party's input too well.

\paragraph{Computing XOR.}

\begin{definition}[$\alpha$-accurate XOR]\label{def:XOR}
	Protocol $\pi = (\Ac,\Bc)$  is computing the XOR functionality in a {\sf $\alpha$-correct manner, denoted  $\alpha$-\correct}, \wrt  $\alpha \colon \N \mapsto \R^+$, if  for any $x,y\in \zo$ it holds that  $ \pr{\out_{\pi}^\Ac(1^\pk,x,y)= \out_\pi^\Bc(1^\kappa,x,y)=x\xor y }  \ge \frac12 +\alpha(\kappa)$.
	
	Such protocols are   {\sf symmetric}, if the parties always agree on the output (\ie $\out_\pi^\Ac(1^\kappa,x,y)= \out_\pi^\Bc(1^\kappa,x,y)$).
\end{definition}
We focus on symmetric protocols with constant $\alpha$ (independent of $\kappa$). 

\paragraph{Our result.}

 \begin{theorem}\label{thm:DPXO}
	Let  $\eps\in [0,1]$.  Assume there exists a symmetric $(21\eps^2)$-\correct,  $(\eps,\eps^3)$-\DP  protocol for computing  XOR,  then there exists an \io key-agreement protocol.
\end{theorem}

The proof is given below. But first note that this range of parameters is indeed achievable by a computationally secure deferentially private protocol.
Consider the functionality $f_{\alpha}(x,y)$ which outputs $x \xor y \xor U_{1/2-\alpha}$ (where $U_{1/2-\alpha}$ is an independent biased coin which is one with probability $1/2 - \alpha$). Assuming OT (oblivious transfer), there exists a two-party protocol that securely implements  $f_{\alpha}$, and this protocol is $\eps$-DP, for $\eps = \Theta(\alpha)$. This is the best possible differential privacy that can be achieved for accuracy $\alpha$. On the other extreme, an $\Theta(\eps^2)$-accurate, $\eps$-differential private,  protocol for computing  XOR  can  be constructed (with information theoretic security) using the so-called \textit{randomized response} approach of \citet{W65a}, as shown in \cite{GoyalMPS2013}. Thus, it is natural to ask what are the minimal computational assumption that are needed for an $\alpha$-accurate, $\eps$-DP computation of XOR, for intermediate choices of $\eps^2 \ll \alpha \ll \eps$.  In this paper, we take a step toward resolve this problem and prove that $\io$ key-agreement is implied for essentially any intermediate $\eps^2 \ll \alpha \ll \eps$.

We now prove \cref{thm:DPXO}.

\begin{proof}
Let $\pi = (\Ac,\Bc)$ be an   $\alpha$-\correct,  $(\eps, \delta)$-\DP  protocol for computing  XOR. We assume \wlg that  $\pi$'s transcript contain the security parameter, so we can omit it from the distinguisher list of inputs. Consider the following no-input protocol $\hPi$.
\begin{protocol}[$\hPi= (\hAc,\hBc)$]
	\item Parameters: security parameter $1^\secParam$.
	
	\item Operation:
	
	\begin{enumerate}
\item 	 $\hAc$ samples  $x\la\zo$ and  $\hBc$ samples  $y\la\zo$.

\item The parties interact in $(\Ac(x),\Bc(y))(1^\kappa)$, with $\hAc$ and $\hBc$ taking the role of $\Ac$ and $\Bc$ respectively.
          Let $\out$ be the (common) output of this interaction.

  \item If $\out = 0$, the parties locally outputs $x$ and $y$ respectively.

       Otherwise, the parties locally outputs $x$ and $1-y$ respectively.
\end{enumerate}
\end{protocol}

Since $\pi$ is symmetric, its $\alpha$ correctness yields that
\begin{align}\label{eq:DPXO:00}
\pr{\out^{\hAc}_\hPi(1^\kappa)= \out^{\hBc}_\hPi(1^\kappa)} \ge \frac12 +\alpha
\end{align}

Since $\pi$ is symmetric  and $(\eps,\eps^3)$-\DP, the \emph{output} of each party of $\hPi$ is  $(\eps,\eps^3)$ differential private from the other party. Namely,  for any \ppt distinguisher $\Dc$ and uniformly chosen bit $X$,
\begin{align}\label{eq:DPXO:01}
\pr{\Dc(1^\kappa,\view^{\hAc}_\hPi(1^\kappa),\out^\hBc_\hPi(1^\kappa))= 1}  \in  e^{\pm \eps} \cdot \pr{\Dc(1^\kappa,\view^{\hAc}_\hPi(1^\kappa),X)= 1} \pm \delta
\end{align}
for all but finitely many $\kappa$'s,  and similarly for the output of $\hAc$.

Let $\rho= \eps^3$. By \cref{thm:main},  either $\hPi$ can be used to construct an io key-agreement protocol, or  it is  io-$\rho$-uncorrelated.  Since we would like to prove the former,  we assume that the latter  holds and derive a contradiction for the assumed combination of privacy and accuracy of $\pi$.

Since protocol $\pi$ is io-$\rho$-uncorrelated, there exists a \pptm  (decorrelator) \Decr that outputs a pair of numbers in $[0,1]$, and an infinite set $\I\subseteq \N$ such that the following holds:
Let $Z=\set{Z_\pk=(X_\kappa,Y_\kappa,T_\kappa,R_\pk)= (\out^\hAc_\hPi(1^\kappa),\out^\hBc_\hPi(1^\kappa),\trans_\hPi(1^\kappa),R_\pk)}_{\pk\in\N}$, where $R_\kappa$ is the uniform string whose length bounds the number of coins used by $\Decr$ on input $t\in \Supp(T_\kappa)$, and let $Z'=\set{Z'_\pk=(U_{p},T_\kappa,R_\kappa)_{p\la \Decr(T_\kappa;R_\kappa)}}_{\kappa\in\N}$. It holds that,
\begin{align}\label{eq:DPXO:000}
Z \cindist_{\rho,\I} Z'
\end{align}
for $U_{p = (p_1,p_2)}$ being the output of two  independent coins, first coin taking the value one with probability $p_1$, and the  second with probability  $p_2$.
Namely, given the transcript and the decorrelator's coins, it is impossible to distinguish too well the  parties' output  from  the  pair of independent coins sample according to the predictor  prediction.

We call  a pair $(p_1,p_2) \in [0,1]^2$   \emph{private}, if $p_1,p_2 \in \frac12  \pm 3\eps$.  Similarly,  $\Decr$ is \emph{private on $\kappa$}, denoted $\kappa$-\private, if
\begin{align*}\label{eq:DPXO}
\pr{\Decr(T_\kappa;R_\kappa) \text{ is  private}} \ge 1- \eps^2
\end{align*}
We use  the privacy of $\hPi$ to derive the following fact.
\begin{claim}\label{claim:DPXO}
$\Decr$ is  $\kappa$-\private for all   but finitely many  $\kappa\in \I$.
\end{claim}
The proof of  \cref{claim:DPXO} is given below, but we first  use it to conclude the theorem's proof.  Let  $\kappa\in \I$  be such that $\Decr$ is $\kappa$-\private.
It follows that
\begin{align*}
\ppr{(p_1,p_2) \la \Decr(T_\kappa;R_\kappa) }{U_{p_1}=U_{p_2}} \le \frac12  + 18\eps^2 +\eps^2 = \frac12 + 19\eps^2
\end{align*}
where $U_p$ is a uniform coin that takes the value one with probability $p$. By \cref{eq:DPXO:000}, for large enough $\kappa\in \I$ it holds that
\begin{align*}
\pr{X_\kappa=Y_\kappa} &\le \frac12 + 19\eps^2 + \rho\\
&\le \frac12 + 20\eps^2 \\
&< \frac 12 + \alpha ,
\end{align*}
in contradiction to \cref{eq:DPXO:00}.
\end{proof}

\begin{proof}[Proof of \cref{claim:DPXO}]
	
For $\kappa\in \I$ for which  $\Decr$ is not $\kappa$-\private, assume \wlg that  $\beta = \pr{\Decr(T_\kappa;R_\kappa)_1 \ge \frac12  + 3\eps} \ge \eps^2$. Consider the distinguisher $\Dc$ that on input $(p,x)$, outputs one if $p \ge \frac12  + 3\eps$ and $x=1$.  By assumption
\begin{align*}
\ppr{p  \la \Decr(T_\kappa;R_\kappa)_1;x \la U_p}{\Dc(p,x) = 1} \ge\beta\cdot (\frac12 +3\eps ) = \frac{\beta}2 +3\beta\eps
\end{align*}
  Hence, \cref{eq:DPXO:000} yields that large enough $\kappa\in \I$,
 \begin{align*}
 \ppr{p  \la \Decr(T_\kappa;R_\kappa)_1;x \la X_\kappa}{\Dc(p,x) = 1}  \ge  \frac{\beta}2 + 3\beta\eps- \rho \ge  \frac{\beta}2 +2\beta\eps
\end{align*}
Since,
 \begin{align*}
\ppr{p  \la \Decr(T_\kappa;R_\kappa)_1;x \la U_{1/2}}{\Dc(p,x) = 1} =  \beta/2 ,
\end{align*}
we conclude that
 \begin{align*}
\pr{\Dc( \Decr(T_\kappa;R_\kappa)_1,X_\kappa) = 1} & \ge \frac\beta 2 +2\beta\eps  = \frac{\beta}2 (1+ 2\eps)  + \beta\eps\\
& > e^{\eps} \cdot \frac\beta2+  \beta\eps\\
& = e^{\eps} \cdot \pr{\Dc( \Decr(T_\kappa;R_\kappa)_1,U) = 1}+   \beta\eps\\
&\ge e^{\eps} \cdot \pr{\Dc( \Decr(T_\kappa;R_\kappa)_1,U) = 1}+   \eps^3.
\end{align*}
Namely, the algorithm that on input $(t,x)$ samples an independent uniform string $r$, and returns $\Dc(\Decr(t;r)_1,x)$, contradicts the assumed differential privacy  of $\pi$ (see \cref{eq:DPXO:01}).
\end{proof}

 \subsection{External Differential Privacy}\label{sec:EDP}
 Our  result  extends to a weaker notion of differential  privacy, that only  guarantee to hold  against external observers.
 \begin{definition}[$(\eps,\delta)$-external differential privacy]\label{def:EDP}
 	A single-bit input two-party protocol $\pi = (\Ac,\Bc)$ is {\sf $(\eps,\delta)$-external differentially private, denoted $(\eps,\delta)$- \EDP},  \wrt  $\eps,\delta \colon \N \mapsto \R^+$, if  for  any \ppt distinguisher $\Dc$ and $x,y,y'\in\zo$, for all but finitely many $\kappa$'s it holds that
 	
 	$$ \pr{\Dc(1^\kappa,\trans_\pi(1^\kappa,x,y))= 1}  \in  e^{\pm \eps(\kappa)} \cdot \pr{\Dc(1^\kappa,\trans_\pi(1^\kappa,x,y'))= 1}  \pm \delta(\kappa)$$
 	for all but finitely many $\kappa$'s, and same for $\Bc$'s input.
 \end{definition}
 Namely,  privacy is only required to hold against an external viewer that sees only the protocol transcript. Achieving  external privacy is  typically much simpler than the full-fledged notion of \cref{def:DP}.  In particular, functionalities such as XOR (with external privacy) can by implemented using key-agreement protocols, this is in contrast to the full-fledged notion of differential privacy \cref{def:DP} that requires oblivious transfer (as was recently shown in \cite{HMSS}).   
 
 we can construct differ for privacy definition  such protocol only need to assume  key-agreement, where we currently only know how to construct them assuming  oblivious transfer require for the  full-fledged notion.

 A protocol has \emph{explicit output} if the parties' common   output appears explicitly in the transcript. For such protocols we have the following result.
 \begin{theorem}\label{thm:EDPXO}
 	Let  $\eps\in [0,1]$.  Assume there exists an explicit-output    $(21\eps^2)$-\correct,  $(\eps,\eps^3)$-\EDP  protocol for computing  XOR,  then there exists an \io key-agreement protocol.
 \end{theorem}
 \begin{proof}
 Follows the same line as the proof of \cref{thm:EDPXO}.
 \end{proof}

\section{Conclusion and Open Problems}

In this paper, we prove a dichotomy theorem (\cref{thm:main}) for \ppt two-party protocols with no inputs and single bit outputs: every such protocol is either $\rho$-uncorrelated (for every $\rho>0$, on infinitely many $\kappa$'s) or it implies key agreement (on infinitely many $\kappa$'s). The theorem comes with caveats: it has ``infinitely many $\kappa$'s'' in both statements (rather than just in one), and it only achieves constant $\rho>0$. A natural open problem is to remove these caveats from \cref{thm:main} (it is natural  to first try and  remove the caveats from \cref{thm:Sim} and \cref{thm:forecaster}).

In this paper, we only discuss protocols where each party outputs a single bit. Our results can be extended to the case that each party outputs a number of bits that is constant (and does not grow with the security parameters). We point out that our results on simulators and forectasters do not extend to the case where the number of bits that each party outputs is large. More specifically, assuming the existence of one-way functions, there do not exist simulators or forecasters for such protocols, and this is the case even if we ignore the second party and only focus on simulating the output $X$ of the first party $A$, given the transcript $T$. In order to see this, consider the case that $T=f(X)$, for a one-way function $f$, where $X$ is uniformly chosen by the party $A$, and is also her output in the protocol. By the security of the one-way function, it is impossible for a polynomial time simulator that is given $T=f(X)$ to output $X'$ such that that the pair $(T,X)$ is computationally indistinguishable from the pair $(T,X')$. This shows that the existence of simulators and forecasters that is guaranteed  in  \cref{thm:Sim,thm:forecaster}, does not hold for protocols where the outputs of the parties is long. It is natural to ask whether some form of a dichotomy theorem applies for protocols that output many bits.

Other interesting open problems are related to our applications. What is the minimal assumption needed for differentially private computation of natural functions? (This question can be asked for various ranges of accuracy and differential privacy parameters). 
For the differentially private XOR functionality this question was fully resolved by the resent  subsequent work of Haitner, Mazor, Shaltiel and Silbak \cite{HMSS}. Where they showed that any non-trivial $\XOR$ can be used to construct an oblivious transfer protocol (without the infinitely often), moreover, this result also applies for sub constant leakage and accuracy.

Can the coin tossing result of \cite{HMO} be extended to hold for a number of rounds that depends on the security parameter?

\bibliographystyle{abbrvnat}
\bibliography{KADichotomy}


\end{document}